\documentclass[12pt]{article}
\usepackage[margin=1in]{geometry}
\usepackage{amsmath,amssymb,amsthm,mathtools}
\usepackage{booktabs}
\usepackage{longtable}
\usepackage{graphicx}
\usepackage{float}
\usepackage{natbib}
\setcitestyle{round}
\usepackage{hyperref}
\usepackage{bm}
\usepackage{xcolor}
\hypersetup{
    colorlinks=true,
    linkcolor=blue,
    urlcolor=blue,
    citecolor=blue
}
\usepackage{comment}
\usepackage{setspace}
\usepackage{microtype}
\newtheorem{theorem}{Theorem}
\newtheorem{proposition}{Proposition}
\newtheorem{lemma}{Lemma}

\newtheorem{assumption}{Assumption}
\newtheorem{condition}{Condition}
\theoremstyle{definition}

\newcommand{\R}{\mathbb{R}}
\newcommand{\E}{\mathbb{E}}

\newcommand{\Var}{\text{Var}}
\newcommand{\tr}{\operatorname{tr}}
\newcommand{\diag}{\text{diag}}

\newcommand{\bZ}{\mathbf{Z}}

\DeclareMathOperator*{\argmin}{arg\,min}

\title{Double Descent and Benign Overfitting in \\ Macroeconomic Forecasting}

\author{Andrea Carriero\thanks{Queen Mary University 
of London. Email: \texttt{a.carriero@qmul.ac.uk}} 
\and Florian Huber\thanks{University of Salzburg. Email: \texttt{florian.huber@plus.ac.at}} 
\and Davide Pettenuzzo\thanks{Brandeis University. Email: \texttt{dpettenu@brandeis.edu}}}
\date{}

\begin{document}
\maketitle

\begin{abstract}
{We study double descent and benign overfitting in macroeconomic forecasting. We document that double-descent risk curves arise in standard macroeconomic datasets that are driven by a small number of latent factors, and we characterize when the underlying benign-overfitting mechanism holds. The conditions of \citet{bartlett2020benign} are satisfied under the exact factor model and can also hold under the more realistic approximate factor model, provided idiosyncratic variances are not too dispersed across series. Because macroeconomic panels have only moderate dimensions, the overparameterization ratio $N/T$ required by the theory is not naturally available. Our solution is to augment the data with synthetic copies from an estimated factor model and we prove that this strategy converges to a kernel ridge regression with a factor-structured kernel. Using monthly (FRED-MD) and quarterly (FRED-QD) US data, the resulting estimator consistently outperforms the Stock-Watson factor model for point forecasting across all series and horizons, with gains that are pervasive, statistically significant, and increasing with the forecast horizon. Our results suggest that benign overfitting, when it works, succeeds because overparameterization implicitly constructs a well-behaved kernel, not because overparameterization is intrinsically desirable.}

\end{abstract}

\clearpage
\onehalfspacing

\section{Introduction}\label{sec:intro}

A striking finding in modern machine learning theory is that models that fit the training data perfectly can nonetheless forecast well out of sample. As the number of parameters grows, the forecast error traces a characteristic \emph{double-descent} curve, first spiking at the interpolation threshold $N=T$ and then falling again, sometimes below the level attained by the best classical estimator.  This phenomenon, known as \emph{benign overfitting}, defies the classical bias-variance tradeoff intuition: more parameters than observations does not necessarily imply worse out-of-sample performance \citep{belkin2019reconciling}. The conditions under which this occurs are now well understood in the linear regression case, as established by \citet{bartlett2020benign} (henceforth, BLLT). The key requirement is a specific factor structure on the predictors: a small number of strong principal components must capture most of the signal, while the remaining components each explain a similar, negligible share of the variance. When this residual spectrum is sufficiently flat, the noise absorbed by the interpolating estimator is spread thinly and evenly across many components, and its contribution to forecast error becomes negligible.

Macroeconomic datasets provide a natural setting to test these conditions. A large body of work in empirical macroeconomics has established that panels of macroeconomic and financial indicators are well approximated by factor models, in which a small number of latent common factors drive the co-movement of many observable series \citep{stock2002forecasting, stock2002macroeconomic, bai2002determining}. The implied covariance structure, a few spiked eigenvalues from the factors plus a residual idiosyncratic component, bears a strong resemblance to the spectral conditions that the benign overfitting literature requires. Whether these conditions hold in practice depends on the dimensionality of the panel: the noise dilution argument above requires the number of small principal components, roughly $N$ (the number of predictors in the panel) minus the number of factors, to be large relative to the sample size $T$. However, existing macroeconomic panels are only moderately high-dimensional, with typical cross-sectional dimensions of $N \sim 100$--$200$ and time series lengths of $T \sim 200$--$700$. Whether these conditions are met in practice, and whether the moderate dimensions of macroeconomic panels allow them to be exploited for forecast gains, are the central questions this paper addresses. 

Our paper makes the following contributions. First, we provide the first illustration of the double-descent risk curve in macroeconomic forecasting, focusing on four representative macroeconomic series: CPI inflation, real personal income, industrial production, and the unemployment rate. In particular, we show how, in each case, the mean squared error (MSE) changes as the effective number of parameters in the forecasting model increases, starting from a simple AR(1) model and adding an increasing number of principal components until reaching the ``overparameterized'' regime, pushing the effective parameter count beyond the saturation/interpolation point where the number of parameters equals the sample size. A characteristic double-descent pattern emerges robustly across these four series: the MSE rises sharply near the interpolation point, then falls and stabilizes in the overparameterized regime.

Second, having documented that the double-descent pattern arises in macroeconomic data, we move to study when the factor structure of macroeconomic panels formally satisfies the conditions for benign overfitting. The theoretical results of \citet{bartlett2020benign} were derived under i.i.d. data, and regardless of the dimension of the panel, it is not obvious that they carry over to the serially dependent, factor-structured data typical of macroeconomic panels. We show that they do, under both the exact and approximate factor model assumptions commonly maintained in the macroeconomic literature. Under the exact factor model, where idiosyncratic errors are cross-sectionally uncorrelated, the first two BLLT conditions hold automatically as consequences of the factor model assumptions, and the third condition holds whenever the idiosyncratic variances are sufficiently similar across series. Under the approximate factor model, which permits residual cross-sectional correlations, the first two conditions still hold, but the third becomes harder to satisfy: correlations among the idiosyncratic components can cause their variances to become more dispersed across series, making benign overfitting less likely. 

Third, our first two contributions raise a practical question: if the BLLT conditions are satisfied in macroeconomic panels, can we actually exploit them given that these panels have only moderate dimensions ($N \sim 100$--$200$, $T \sim 200$--$700$), so that the BLLT regime $N/T \to \infty$ does not hold directly? We address this by proposing to augment the predictors with $B$ synthetic copies of the data generated using a factor model trained on the original time series, artificially increasing the effective $N/T$ ratio. The key insight is that each synthetic copy preserves the estimated factor structure while drawing fresh idiosyncratic noise, so that two time periods with similar factor values produce similar synthetic predictors; in the limit, this implicit similarity measure is precisely a kernel. We prove that, as the number of synthetic replications $B \to \infty$, this procedure converges to kernel ridge regression with a factor-structured kernel that measures similarity between time periods through their estimated factor realizations, and a ridge parameter controlled by the total idiosyncratic variance. Thus, the synthetic expansion works not by adding information, but by constructing an implicit kernel whose structure is inherited from the factor model.

These contributions together support a central conclusion: benign overfitting, when it works, succeeds because overparameterization implicitly constructs a well-behaved kernel, not because overparameterization is intrinsically desirable. 

We then turn to verify this claim empirically, evaluating how a semiparametric factor kernel model (obtained by combining a linear AR specification with a nonparametric kernel component) fares against the standard factor model benchmark \citep{stock2002forecasting, stock2002macroeconomic}. We perform this comparison using all available series in the FRED-MD dataset of \citet{mccracken2016fred}, forecasting $h = 1, 3, 6, 12$ months ahead using a rolling-window setup where the AR(1) dynamics is estimated via generalized least squares (GLS). We find that the factor kernel consistently outperforms the factor model benchmark across all series and forecast horizons. The gains are pervasive --- the factor kernel beats the benchmark for 92--99\% of monthly targets and 96--99\% of quarterly targets --- and become increasingly large and statistically significant as the forecast horizon lengthens. At short horizons, improvements are modest but highly consistent; at longer horizons, more than half of the wins are statistically significant by the Diebold--Mariano test. We interpret these findings as empirical evidence that the semiparametric factor kernel provides regularization naturally adapted to the data's factor structure, connecting the modern interpolation literature to the classical tradition of factor-based forecasting \citep{stock2002forecasting, stock2002macroeconomic} and to kernel methods in economics.

This paper contributes to multiple strands of literature. First, we add to the recent empirical literature that has documented how ML methods can outperform standard benchmarks for macroeconomic forecasting \citep{medeiros2021forecasting, coulombe2022machine, chi2025macroeconomic}. We complement the existing work by identifying the spectral mechanism that determines when and why such gains can manifest. Our kernel-based approach is related to the nonlinear and kernel methods applied to factor-augmented forecasting by \citet{exterkate2016nonlinear} and \citet{boot2019macroeconomic}, though our kernel arises endogenously from the factor structure rather than being imposed exogenously. We also add to the literature on the theory behind benign overfitting. \citet{bartlett2020benign} establish sufficient and necessary conditions under i.i.d.\ data; \citet{nakakita2022benign} extend the analysis to dependent processes; \citet{hastie2022surprises} and \citet{bunea2021interpolating} provide further characterizations in related settings. We contribute to this theory by mapping the BLLT conditions onto the factor structure of macroeconomic panels, characterizing when they hold and when they fail, and deriving an implementable kernel estimator that achieves benign overfitting when the native $N/T$ ratio is insufficient. The factor model forecasting tradition of \citet{stock2002forecasting, stock2002macroeconomic} provides both our empirical benchmark and the structural motivation for the factor kernel. Our work builds on the static factor model framework; the dynamic factor model literature \citep{forni2000generalized, forni2005generalized} and the targeted predictor approach of \citet{bai2008forecasting} offer complementary perspectives on dimension reduction in large panels. Finally, our data augmentation strategy is related in spirit to the Bayesian shrinkage approach of \citet{demol2008forecasting}, which also exploits the factor structure for regularization, though through a different mechanism.

The remainder of the paper is organized as follows. Section~\ref{sec:setup} introduces the ridgeless estimation framework and characterizes the double-descent phenomenon. Section~\ref{sec:bllt_conditions} reviews the BLLT conditions for benign overfitting and their extension to dependent data. Section~\ref{sec:factor_bartlett} establishes when factor-structured data satisfy the BLLT conditions, showing that the first two conditions hold automatically while the third depends on the dispersion of idiosyncratic variances. Section~\ref{sec:augmentation_andRKHS} develops the synthetic factor augmentation strategy and proves its equivalence to kernel ridge regression with a factor-structured kernel. Section~\ref{sec:empirical} presents the out-of-sample forecasting evaluation comparing the factor kernel to the Stock-Watson benchmark. Section~\ref{sec:conclusion} concludes. The appendices contain proofs (Appendix~\ref{sec:appendix}), the double-descent analysis using quarterly FRED-QD data (Appendix~\ref{app:fred_qd}), an illustration of bias and variance behavior under alternative spectral shapes (Appendix~\ref{sec:bias-spectra}), a derivation of the relationship between tail index and concentration ratio (Appendix~\ref{app:pareto_tail}), and the GLS formulation of the semiparametric kernel regression (Appendix~\ref{app:rkhs_formulation}).

\section{Double-Descent and Macroeconomic Forecasting}\label{sec:setup}

Macroeconomic forecasters routinely face many more candidate predictors than observations. Standard datasets such as FRED-MD contain over 120 monthly series against sample sizes of $T \sim 400$--$800$, and including all series directly in a forecasting regression inflates estimation variance and degrades out-of-sample performance. The literature's response has been to reduce the effective parameter count: extract a small number of common factors \citep{stock2002forecasting, stock2002macroeconomic}, shrink coefficients via ridge regression or the LASSO \citep{medeiros2021forecasting, coulombe2022machine}, or select variables via information criteria or Bayesian model averaging. All of these approaches keep the number of estimated parameters well below $T$, appealing to the familiar bias-variance tradeoff to justify parsimony.

Inspired by recent developments in the machine learning literature, this paper investigates a different approach: rather than shrinking toward more parsimonious models, we ask what happens when the number of regressors $N$ is pushed past the \emph{interpolation threshold} $N = T$ (the point at which the model can fit every observation exactly) and into the overparameterized regime. We formalize this setup below and characterize the conditions under which this approach can deliver good out-of-sample forecasts.

Let $\mathbf{x}_t \in \R^N$ be a vector of macroeconomic predictors at date $t$ and let $y_t$ be a scalar target series. We assume that the forecasting model is linear:
\begin{equation}\label{eq:model}
	y_t = \mathbf{x}_t'\boldsymbol{\beta}^* + \varepsilon_t, \qquad t = 1, \ldots, T,
\end{equation}
where $\boldsymbol{\beta}^* \in \R^N$ is the vector of population projection coefficients and $\varepsilon_t$ is the forecast error, with $\E[\varepsilon_t \mid \mathbf{x}_t] = 0$ and $\E[\varepsilon_t^2] = \sigma^2$. In matrix notation, 
\begin{equation}\label{eq:model_stacked}
	\mathbf{y} = \mathbf{X}\boldsymbol{\beta}^* + \boldsymbol{\varepsilon}
\end{equation}
where  $\mathbf{y} \in \R^{T}$ and $\mathbf{X} \in \R^{T \times N}$ stacks the $T$ observations row-wise. 

\subsection{Ridgeless estimation and identification}

The behavior of the least-squares estimator depends critically on the relationship between $N$ and $T$. When $N < T$ (\emph{under-parameterized} regime), $\mathbf{X}$ has full column rank, $\mathbf{X}'\mathbf{X}$ is invertible, and OLS delivers the unique unbiased estimator $\hat{\boldsymbol{\beta}} = (\mathbf{X}'\mathbf{X})^{-1}\mathbf{X}'\mathbf{y}$. The model is identified in the usual sense and the standard bias-variance tradeoff applies. When $N = T$ (\emph{interpolation point}), $\mathbf{X}$ is square and, if full rank, the system $\mathbf{X}\boldsymbol{\beta} = \mathbf{y}$ has a unique solution that fits the training data with zero residuals, $\hat{\boldsymbol{\beta}} = \mathbf{X}^{-1} \mathbf{y}$. In contrast, when $N > T$ (\emph{overparameterized} regime), the system $\mathbf{X}\boldsymbol{\beta} = \mathbf{y}$ is underdetermined: in other words, there are infinitely many coefficient vectors $\hat{\boldsymbol{\beta}}$ that produce a perfect in-sample fit, and since $\mathbf{X}'\mathbf{X}$ is rank-deficient standard OLS cannot be applied directly.

Among the infinitely many solutions that perfectly fit the training data when $N > T$, a natural choice is the one that solves the normal equations with the smallest coefficient norm, i.e. the \emph{minimum norm estimator}:
\begin{equation}\label{eq:min_norm}
    \hat{\boldsymbol{\beta}} = \argmin_{\boldsymbol{\beta}} \lVert\boldsymbol{\beta}\rVert^2 \quad \text{subject to} \quad \mathbf{X}'\mathbf{X}\boldsymbol{\beta} = \mathbf{X}'\mathbf{y},
\end{equation}
which has closed form solution
\begin{equation}
    \boldsymbol{\hat{\beta}}=(\mathbf{X}'\mathbf{X})^{\dagger}\mathbf{X}'\mathbf{y}=\mathbf{X}'(\mathbf{X}\mathbf{X}')^{\dagger}\mathbf{y},
\end{equation}
where $(\mathbf{X}'\mathbf{X})^{\dagger}$ is the pseudoinverse of $(\mathbf{X}'\mathbf{X})$. This estimator turns out to correspond, in the limit, to the ridge estimator $\hat{\boldsymbol{\beta}}_\lambda = (\mathbf{X}'\mathbf{X} + \lambda \mathbf{I}_N)^{-1}\mathbf{X}'\mathbf{y}$ as the penalty parameter $\lambda$ approaches zero. For this reason, it is also referred to as the \emph{ridgeless estimator}.

Depending on whether $\mathbf{X}$ is full row or full column rank, we end up with the following three cases:
\begin{equation}\label{eq:ridgeless}
	\hat{\boldsymbol{\beta}} = \begin{cases}
		(\mathbf{X}'\mathbf{X})^{-1}\mathbf{X}'\mathbf{y}, & N < T, \\[4pt]
		\mathbf{X}^{-1}\mathbf{y},                         & N = T, \\[4pt]
		\mathbf{X}'(\mathbf{X}\mathbf{X}')^{-1}\mathbf{y}, & N > T.
	\end{cases}
\end{equation}

The underparameterized ($N<T$) and overparameterized ($N>T$) regimes imply a qualitatively different behavior for the mean squared forecast error (MSFE), and the transition from one regime to the next can give rise to the double-descent pattern we mentioned in the Introduction. We will characterize this formally in Section~\ref{sec:double_descent}. Before turning to that, we first examine what happens as we transition across the interpolation threshold ($N = T$) in terms of the identifiability of $\boldsymbol{\beta}^*$.

As it is well known, when $N < T$ the model is identified in the usual sense: $\mathbf{X}$ has full column rank $N$, $\mathbf{X}'\mathbf{X}$ is invertible and OLS uniquely determines $\hat{\boldsymbol{\beta}}$, with $\E[\hat{\boldsymbol{\beta}} \mid \mathbf{X}] = \boldsymbol{\beta}^*$. The estimator is unbiased and $\boldsymbol{\beta}^*$ is point-identified. The interpolation threshold marks a fundamental change in identification. When $N>T$, the data matrix $\boldsymbol{X}$ has full row rank $T$, it admits only a right inverse and can no longer uniquely identify $\boldsymbol{\beta}^*$. The minimum norm estimator will be biased: 
\begin{equation}\label{eq:bias_over}
\E[\boldsymbol{\hat{\beta}}\mid \boldsymbol{X}]=\bm{P}_r \boldsymbol{\beta}^*,
\end{equation}
where $\bm{P}_r = \boldsymbol{X}'(\boldsymbol{X}\boldsymbol{X}')^{-1}\boldsymbol{X}$ is the orthogonal projector onto the row space of $\boldsymbol{X} \subset \R^N$. Intuitively, with $T$ observations and $N > T$ predictors, the training data can pin down at most $T$ linearly independent combinations of the $N$ parameters. The remaining $N - T$ combinations fall in the null space of $\mathbf{X}$ and are invisible to any estimator. $\boldsymbol{\beta}^*$ is therefore set-identified and the identified set is the affine subspace

\[
\mathcal{I}(\boldsymbol{\beta}^*) = \{\tilde{\boldsymbol{\beta}} \in \R^N : \mathbf{X}\tilde{\boldsymbol{\beta}} = \mathbf{X}\boldsymbol{\beta}^*\} = \boldsymbol{\beta}^* + \mathrm{null}(\mathbf{X}),
\]
The data can only help identify the component of $\boldsymbol{\beta}^*$ in the row space of $\mathbf{X}$, which spans the $T$ directions that can be inferred from the estimation sample. The remaining $N - T$ directions, lying in the null space of $\mathbf{X}$, are unrecoverable regardless of the available sample size. 
We can then decompose $\boldsymbol{\beta}^*$ as\footnote{The null-space membership $\boldsymbol{\beta}^*_\perp\in\mathrm{null}(\mathbf{X})$ follows from $\mathbf{X}(\mathbf{I}-\mathbf{P}_r) = \mathbf{X} - \mathbf{X}\mathbf{X}'(\mathbf{X}\mathbf{X}')^{-1}\mathbf{X} = \mathbf{X} - \mathbf{X} = \mathbf{0}$.}
\[
\boldsymbol{\beta}^* = \mathbf{P}_r\boldsymbol{\beta}^* + \boldsymbol{\beta}^*_{\perp},
\qquad
\boldsymbol{\beta}^*_{\perp} := (\mathbf{I} - \mathbf{P}_r)\boldsymbol{\beta}^* \in \mathrm{null}(\mathbf{X}),
\]
The minimum-norm estimator recovers $\mathbf{P}_r\boldsymbol{\beta}^*$ and sets the null-space component to zero,\footnote{Substitute $\mathbf{y}=\mathbf{X}\boldsymbol{\beta}^*+\boldsymbol{\varepsilon}$ into $\hat{\boldsymbol{\beta}}=\mathbf{X}'(\mathbf{X}\mathbf{X}')^{-1}\mathbf{y}$ to get $\hat{\boldsymbol{\beta}} = \mathbf{P}_r\boldsymbol{\beta}^* + \mathbf{X}'(\mathbf{X}\mathbf{X}')^{-1}\boldsymbol{\varepsilon}$; taking $\E[\,\cdot\mid\mathbf{X}]$ yields $\E[\hat{\boldsymbol{\beta}}\mid\mathbf{X}]=\mathbf{P}_r\boldsymbol{\beta}^*$.} so its bias is precisely the unidentified component:
\[
\E[\hat{\boldsymbol{\beta}} \mid \mathbf{X}] - \boldsymbol{\beta}^* = -\boldsymbol{\beta}^*_{\perp}.
\]

Since every element of $\mathcal{I}(\boldsymbol{\beta}^*)$ produces the same in-sample fit and the same forecasts for test points within the training span, any selection rule that picks one element over another is essentially an identifying restriction on the unidentified component. The ridgeless estimator sets all unidentifiable coefficient directions to zero:
\[
(\mathbf{I} - \mathbf{P}_r)\hat{\boldsymbol{\beta}} = \mathbf{0}.
\]
The choice of the selection rule has no consequence for in-sample fit, since the fitted values are an unbiased estimator of the conditional mean of $Y$: $\mathbb{E}[\boldsymbol{X}\boldsymbol{\hat{\beta}}\mid \boldsymbol{X}]=\boldsymbol{X}\bm{P}_r\boldsymbol{\beta}=\boldsymbol{X}\boldsymbol{\beta}$; It matters only through its effect on the bias component $\boldsymbol{\beta}^*_\perp$ when forecasting future observations whose predictor values fall outside the span of the training data.

\subsection{Double descent}\label{sec:double_descent}
Suppose we forecast a new observation $(\mathbf{x}_0, y_0)$ with $\mathbf{x}_0 \sim (0, \boldsymbol{\Sigma})$ and $y_0 = \mathbf{x}_0'\boldsymbol{\beta}^* + \varepsilon_0$. The MSFE expands as
\begin{equation}\label{eq:test_risk}
    \mathrm{MSFE} = \E\!\big[(\hat{\boldsymbol{\beta}} - \boldsymbol{\beta}^*)'\boldsymbol{\Sigma}(\hat{\boldsymbol{\beta}} - \boldsymbol{\beta}^*)\big] + \sigma^2,
\end{equation}
and decomposes into bias and variance components:\footnote{Write $\hat{\boldsymbol{\beta}} - \boldsymbol{\beta}^* = (\hat{\boldsymbol{\beta}} - \E[\hat{\boldsymbol{\beta}} \mid \mathbf{X}]) + (\E[\hat{\boldsymbol{\beta}} \mid \mathbf{X}] - \boldsymbol{\beta}^*)$ and expand the quadratic form in~\eqref{eq:test_risk}. Taking $\E[\,\cdot \mid \mathbf{X}]$, the cross-term vanishes because $\E[\hat{\boldsymbol{\beta}} - \E[\hat{\boldsymbol{\beta}} \mid \mathbf{X}] \mid \mathbf{X}] = \mathbf{0}$. The first remaining term is $\E_\mathbf{X}[\tr(\boldsymbol{\Sigma}\,\Var(\hat{\boldsymbol{\beta}} \mid \mathbf{X}))]$ (the variance component) and the second is $\E_\mathbf{X}[\|\boldsymbol{\Sigma}^{1/2}(\E[\hat{\boldsymbol{\beta}} \mid \mathbf{X}] - \boldsymbol{\beta}^*)\|^2]$ (the squared bias).}
\begin{equation}\label{eq:risk_bv}
    \mathrm{MSFE} = \underbrace{\E_{\mathbf{X}}\!\left[\|\boldsymbol{\Sigma}^{1/2}(\E[\hat{\boldsymbol{\beta}} \mid \mathbf{X}] - \boldsymbol{\beta}^*)\|^2\right]}_{\text{Bias}^2} + \underbrace{\E_{\mathbf{X}}\!\left[\tr\!\big(\boldsymbol{\Sigma}\,\Var(\hat{\boldsymbol{\beta}} \mid \mathbf{X})\big)\right]}_{\text{Variance}} + \sigma^2.
\end{equation}
Importantly, the $\boldsymbol{\Sigma}$-weighting reflects the fact that forecast accuracy depends not on parameter error per se but on how that error maps into prediction error through the predictor covariance. This is the key to understanding how double descent works. The two components take qualitatively different forms on either side of the interpolation threshold. In the \emph{underparameterized regime} ($N < T$), OLS is unbiased (Bias$^2 = 0$) and
\[
\mathrm{Variance} = \E[\sigma^2\,\mathrm{tr}\!\big(\boldsymbol{\Sigma}(\mathbf{X}'\mathbf{X})^{-1}\big)],
\]
In the \emph{overparameterized regime} ($N > T$), using $\E[\hat{\boldsymbol{\beta}} \mid \mathbf{X}] = \mathbf{P}_r\boldsymbol{\beta}^*$,
\begin{align}
\mathrm{Bias}^2 &=\E[ \|\boldsymbol{\Sigma}^{1/2}\boldsymbol{\beta}^*_\perp\|^2], \label{eq:bias_OP}\\
\mathrm{Variance} &= \E[\sigma^2\,\mathrm{tr}\!\big(\boldsymbol{\Sigma}\,\mathbf{X}'(\mathbf{X}\mathbf{X}')^{-2}\mathbf{X}\big)], \label{eq:var_OP}
\end{align}
where $\boldsymbol{\beta}^*_\perp = (\mathbf{I} - \mathbf{P}_r)\boldsymbol{\beta}^*$ is the unrecovered component defined in the preceding subsection. The derivation of these variance expressions is given in Appendix~\ref{proof:bias_variance}.

The spike in MSFE at the interpolation threshold is a guaranteed consequence of the variance divergence as $N \to T$. Seeing this is easy under the (standard) assumption that the rows of $\bm X$ follow a Gaussian distribution with zero mean and variance $\bm \Sigma$, in which case we can write the variance term in the $N < T$ case as:
\[
    \E[\sigma^2\,\mathrm{tr}\!\big(\boldsymbol{\Sigma}(\mathbf{X}'\mathbf{X})^{-1}\big)] = \sigma^2 \frac{\mathrm{tr}(\bm \Sigma \bm \Sigma^{-1})}{T-N-1} = \sigma^2 \frac{N}{T-N-1}
\]
which diverges as $N \to T-1$. A similar pattern occurs in the $N \to T^+$ case. To see this, again assuming Gaussianity for simplicity, write $\bm X = \bm Z \bm\Sigma^{1/2}$ where $\bm Z$ is $T \times N$ with i.i.d.\ $\mathcal{N}(0,1)$ entries. Since $\bm X\bm X^\top = \bm Z\bm\Sigma \bm Z^\top \preceq \lambda_{\max}(\bm\Sigma)\,\bm Z\bm Z^\top$, we have $(\bm X\bm X^\top)^{-1} \succeq \lambda_{\max}(\bm\Sigma)^{-1}(\bm Z\bm Z^\top)^{-1}$, where $\lambda_{s} (\bm \Sigma)$ (for $s \in \{\text{min}, \text{max} \}$) denotes the minimum/maximum eigenvalue of $\bm \Sigma$.\footnote{We stick to this notation throughout the remainder of the paper} Since $\bm Z\bm Z^\top \sim W_T(N, \bm I_T)$, taking traces and expectations gives the lower bound:
\[
    \E\!\big[\mathrm{tr}\{(\bm X\bm X^\top)^{-1}\}\big] \;\geq\; \frac{1}{\lambda_{\max}(\bm\Sigma)}\,\frac{T}{N-T-1}
\]
for $N > T+1$, which diverges as $N \to T+1$. The bound discards all spectral structure of $\bm\Sigma$ except its largest eigenvalue, but suffices to establish that the variance diverges at the interpolation threshold for any $\bm\Sigma$ with bounded $\lambda_{\max}$. 
At $N = T$, both the underparameterized formula $\sigma^2 \frac{N}{T-N-1}$ and the overparameterized bound $\frac{1}{\lambda_{\max}(\bm\Sigma)}\,\frac{T}{N-T-1}$ are undefined, since both denominators are negative. This is the interpolation threshold at which the double descent peak occurs.

If bias and variance are both bounded as $N/T$ grows beyond the interpolation point, the MSFE will exhibit a \emph{double descent}, that is a pattern in which it reduces twice: a first descent in the underparameterized regime, a spike at the interpolation threshold, and a second descent as the model moves deeper into the overparameterized regime.

Whether the second descent would eventually produce a lower MSFE than the underparameterized regime, i.e. \emph{benign overfitting}, is not guaranteed. To achieve benign overfitting it is not sufficient for the bias and variance terms to be bounded, they have to approach zero, which means the MSFE will have to achieve the noise floor $\sigma^{2}$. As we shall see, this can happen when $\boldsymbol{\Sigma}$ has a spectral shape characterized by a few large eigenvalues and many small, but non-vanishingly small, eigenvalues. In such a situation as $N$ grows, the row space of $\mathbf{X}$ progressively aligns with the high-variance directions, pushing the unrecovered component $\boldsymbol{\beta}^*_{\perp}$ into low-variance tail directions where it contributes little to forecast error. Simultaneously, the training noise gets absorbed thinly across many regressors, reducing the variance term.

\subsection{An empirical illustration}\label{sec:motivation}

Section~\ref{sec:double_descent} established that the MSFE spike at the interpolation threshold is guaranteed, and that a second descent follows whenever the bias term~\eqref{eq:bias_OP} and variance term~\eqref{eq:var_OP} remain bounded as $N$ grows beyond $T$. We now ask whether this pattern is visible in macroeconomic data. Using the FRED-MD dataset, which contains $N = 122$ monthly series spanning 1959:01 to 2025:09 ($T = 798$ observations), we trace the full MSFE curve from the underparameterized region, through the interpolation threshold, and into the overparameterized regime. We split the sample in half, with $T_{\mathrm{train}} = 399$ and $T_{\mathrm{test}} = 399$.\footnote{In this motivating exercise the single-split design is meant to make the double-descent MSFE curve visible. Whether it also delivers systematic forecasting gains is addressed in Section~\ref{sec:empirical}, where we use a full rolling-window out-of-sample design across all targets and horizons. A parallel analysis using the quarterly FRED-QD dataset is presented in Appendix~\ref{app:fred_qd}.} All models include an AR(1) lag as a baseline predictor.

In the notation of equation~\eqref{eq:model_stacked}, $\bm{y}$ is the $T_{\mathrm{train}} \times 1$ vector of the target variable and $\mathbf{X}$ is the $T_{\mathrm{train}} \times N$ matrix stacking the remaining $N = 122$ FRED-MD series as predictors, with each variable standardized using training-window moments and transformed using the FRED-MD provided transformation codes. We consider four target series: CPI inflation (CPIAUCSL), real personal income (RPI),  industrial production (INDPRO), and the unemployment rate (UNRATE). As with the predictors, all four target series were transformed using the FRED-MD transformation codes.

We begin by taking an AR(1) model (one parameter), and gradually increasing model size toward $N = 122$ by adding $r$ principal components (PCs) of $\mathbf{X}$, ranked by eigenvalue, one at a time, where $r$ runs from $1$ to $N$. This yields an ordered path from the AR(1)+1 PC model to the AR(1)+all $N$ PCs model (which is algebraically equivalent to using all $N$ original predictors). The first set of PCs captures the dominant common factors, while later PCs increasingly capture idiosyncratic dynamics. 
\begin{figure}[t!]
    \centering
    \includegraphics[width=1\linewidth]{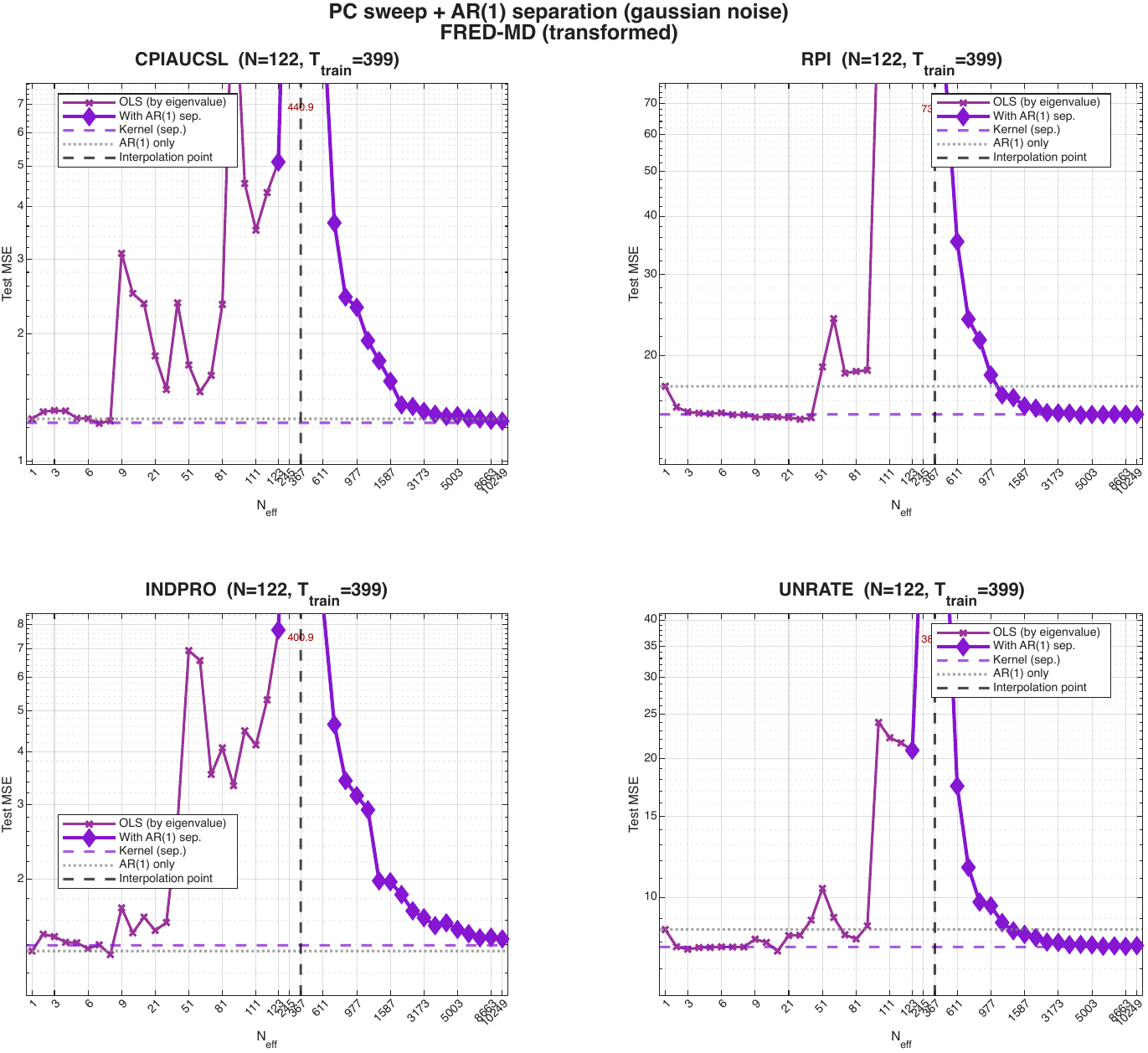}
    \caption{Double-descent MSFE curves for four FRED-MD series.
    The thin purple line traces the PC sweep from $r=1$ to $r=N$;
    the thick purple line traces results based on $B$ synthetic copies of the data $\bm X$.
    The dashed vertical line is the interpolation threshold ($N= T$);
    the purple dashed horizontal line is the asymptotic limit of the estimator, computed from its RKHS (kernel) representation described in Section~\ref{sec:augmentation_andRKHS};
    the gray dotted horizontal line shows the AR(1) baseline.}
    \label{fig:dd_four_panel}
\end{figure}
At $r = N$ the model uses all available variables and we have reached the boundary of the original data. To push into the overparameterized regime (where the number of effective parameters exceeds the number of training observations), we augment the design matrix with synthetic copies of $\bm X$. Each copy preserves the estimated common component (extracted using $\hat{k}$ Bai--Ng factors) but replaces the idiosyncratic errors with fresh Gaussian draws. Appending $B$ such copies creates a design matrix with $N(B+1)$ columns, so the effective number of parameters grows in discrete jumps of $N$. The construction and its theoretical properties are developed in Section~\ref{sec:augmentation_andRKHS}; here we simply use it as a device to extend the MSFE curve beyond the interpolation threshold.

Throughout, $N_{\mathrm{eff}}$ denotes the total number of slope coefficients being estimated (cross-sectional predictors plus the AR(1) lag): $N_{\mathrm{eff}}=r+1$ along the PC path and $N_{\mathrm{eff}}=(B+1)N+1$ with augmentation. The interpolation threshold is reached when $N_{\mathrm{eff}}=T_{\mathrm{train}}$. Since we only have $N=122$ PCs, we start augmenting the dataset with artificial covariates after we have included all PCs. After crossing the interpolation point this augmentation path then extends the analysis into the deeply overparameterized regime.


We apply this procedure to each of the four targets and plot the resulting MSFE against the effective number of parameters. Figure~\ref{fig:dd_four_panel} shows a clear and consistent double-descent pattern across all four series: MSFE rises sharply near the interpolation threshold and then declines in the overparameterized region, confirming that the second descent is a robust feature of macroeconomic data.

The natural next question is whether this decline is deep enough to constitute benign overfitting: in a true out of sample setting, does the MSFE on the right branch eventually approach the noise floor $\sigma^2$, or does it stabilize at a higher level? A rigorous assessment is carried out in Section~\ref{sec:empirical}. Before turning to that, we characterize the theoretical conditions for benign overfitting and show how the factor structure of macro panels naturally satisfies them.

\section{Conditions for Benign Overfitting}\label{sec:bllt_conditions}

This section characterizes the conditions under which benign overfitting occurs. The material draws on the papers cited and is included for completeness, providing necessary background for Section~\ref{sec:factor_bartlett}. Readers already familiar with the benign overfitting literature may skip directly to Section~\ref{sec:factor_bartlett}.

We start with analyzing separately the bias and the variance term in expression (\ref{eq:risk_bv}), and we discuss heuristically how and in which circumstances they may  converge to zero. Then we formally introduce the \citet{bartlett2020benign} (BLLT) conditions for i.i.d. data. Finally, we discuss how the BLLT conditions apply also to dependent data, drawing on the results of \citet{nakakita2022benign}.

Before proceeding, we introduce terminology that will be used throughout. The \emph{spectrum} of the $N \times N$ predictor covariance matrix $\boldsymbol{\Sigma}$ refers to its ordered eigenvalues $\lambda_1 \geq \lambda_2 \geq \cdots \geq \lambda_N > 0$. We partition the spectrum into two regions: the \emph{head} consists of the first $k$ eigenvalues, which are large and capture the dominant directions of variation (in factor models, these correspond to the common factors); the \emph{tail} consists of the remaining $N - k$ eigenvalues, which are small and represent residual variation (in factor models, these correspond to the idiosyncratic components). The split point $k$ is chosen to separate eigenvalues of order $O(N)$ from those of order $O(1)$. The key to benign overfitting is that the head must be low-dimensional (small $k$) while the tail must be both flat (eigenvalues of similar magnitude) and large-dimensional (large $N - k$), so that noise can be dispersed across many weak directions.

\subsection{A heuristic description}
As we discussed above, double descent merely requires that the MSFE eventually decreases past the interpolation point. It says nothing about the level to which it decreases. Benign overfitting is a strictly stronger statement: it requires that the MSFE converge to the noise floor $\sigma^2$ as $N/T \to \infty$, meaning that \emph{both} the bias and the variance components must converge to zero. 

\paragraph{Bias term.} As $N/T \to \infty$ the unrecovered component $\boldsymbol{\beta}^*_{\perp}$ does not vanish in Euclidean norm, which means the minimum norm estimator remains biased. However, the crucial observation is that the bias term in the MSFE rescales the bias by $\boldsymbol{\Sigma}$, i.e., each direction of $\boldsymbol{\beta}^*_{\perp}$ is weighted by the corresponding predictor variance. This means that even though the bias does not vanish in Euclidean norm, the $\boldsymbol{\Sigma}$-weighted bias might, depending on the eigenvalue structure of $\boldsymbol{\Sigma}$ and the alignment of $\boldsymbol{\beta}^*$ with the factor space. 

The key mechanism is as follows. Suppose $\boldsymbol{\Sigma}$ features a strong head of large eigenvalues and a flat tail of small, non-vanishing eigenvalues, and suppose the signal $\boldsymbol{\beta}^*$ is concentrated in the head directions. As $N$ increases, the projection operator $\bm{P}_r$ evolves because $\mathbf{X}$ gains new columns. The row space becomes progressively better aligned with the head space where $\boldsymbol{\beta}^*$ has support, and the residual $(I - \bm{P}_r)\boldsymbol{\beta}^*$ shrinks geometrically. Moreover, since the remaining residual lies in low-variance tail directions with tiny eigenvalues $\lambda_j$, the $\boldsymbol{\Sigma}$-weighted bias $\|\boldsymbol{\Sigma}^{1/2}(\bm I - \bm{P}_r)\boldsymbol{\beta}^*\|^2$ decreases even faster than the unweighted residual.

This mechanism, however, requires that the tail eigenvalues do not vanish too quickly. If they decay exponentially fast, adding new predictors contributes negligible variance to the row space, and $\bm{P}_r$ stops evolving before fully aligning with the head space. In this case, a fixed portion of $\boldsymbol{\beta}^*$ remains unrecovered regardless of $N$, and the bias does not vanish. Conversely, if the tail is too heavy (eigenvalues decay too slowly), the unrecovered component $\boldsymbol{\beta}^*_{\perp}$ may lie in directions with non-negligible variance, preventing the weighted bias from vanishing. 

Thus, benign overfitting requires the tail to be flat: disperse enough (many directions) to allow the row space to expand and cover the signal support, yet light enough (small eigenvalues) to ensure the weighted bias vanishes. Appendix~\ref{sec:bias-spectra} illustrates these requirements through three concrete spectral configurations. In a factor-structured panel, this structure arises naturally. The high-variance head directions correspond to the common factors that drive cross-sectional covariation, while the flat tail corresponds to the idiosyncratic components with similar small variances. If the variable we are forecasting loads predominantly on the factors, then $\boldsymbol{\beta}^*$ is concentrated in the head, $\boldsymbol{\beta}^*_{\perp}$ ends up in the idiosyncratic tail directions, and its contribution to MSFE is negligible even if $\|\boldsymbol{\beta}^*_{\perp}\|$ is not.\footnote{This requires the predictive signal $\boldsymbol{\beta}^*$ to have significant support in fewer than $T$ directions. If it is spread across more than $T$ directions, the $T$-dimensional row space of $\mathbf{X}$ cannot cover them all, and a fixed portion of the signal remains unrecovered regardless of $N$.}

\paragraph{Variance term.} The variance component measures how much in-sample noise is amplified into out-of-sample forecast error by the minimum-norm interpolator. When $N > T$, the interpolator fits the training noise exactly but distributes it across all $N$ coefficient directions. Here, again, the eigenvalues of $\boldsymbol{\Sigma}$ play a key role. 

If $\boldsymbol{\Sigma}$ has a flat tail of small but non-vanishing eigenvalues, then as $N$ grows, the minimum-norm solution distributes the fit over more columns, keeping the $\ell_2$ norm of the coefficient vector small. Since each tail direction contributes little variance, the noise gets absorbed thinly across many regressors, and the out-of-sample variance decreases. Conversely, if the tail is concentrated on a few large eigenvalues, noise accumulates in high-variance directions and the variance remains large. In a factor-structured panel, the idiosyncratic components provide precisely this flat spectral tail: many series with similar small variances across which the interpolated noise is spread evenly. 

\subsection{BLLT conditions}
The BLLT conditions formalize exactly what is the spectral shape of $\bm{\Sigma}$ for which benign overfitting may happen: conditions (i) and (ii) ensure the spectrum has a well-defined factor/idiosyncratic head-tail structure with the signal confined to the head, while condition (iii) ensures the idiosyncratic tail is flat enough to disperse noise. We now state these conditions formally, following \citet{bartlett2020benign}. Let $\lambda_1 \geq \lambda_2 \geq \cdots \geq \lambda_N > 0$ be the eigenvalues of $\bm{\Sigma}$. Define two measures of effective rank for the tail beyond coordinate $k$:
\begin{equation}\label{eq:effective_ranks}
    r_k(\bm{\Sigma}) := \frac{\sum_{i>k} \lambda_i}{\lambda_{k+1}}, \qquad
    R_k(\bm{\Sigma}) := \frac{\big(\sum_{i>k} \lambda_i\big)^2}{\sum_{i>k} \lambda_i^2}.
\end{equation}
The first, $r_k$, measures the total tail variance relative to the largest tail eigenvalue; the second, $R_k$, measures how uniform the tail is. When all tail eigenvalues are equal, $r_k = R_k = N - k$. When the tail is concentrated on a few directions, $R_k \ll r_k$.
Note that both $r_k$ and $R_k$ are scale-invariant: they depend only on the ratios among the tail eigenvalues. Consequently, the BLLT ratios below are unaffected by an overall rescaling of $\bm{\Sigma}$.

Define the \emph{split index}
\[
    k^* := \min\big\{k \geq 0 : r_k(\bm{\Sigma}) \geq b\,T\big\},
\]
for a universal constant $b > 0$. The split index identifies the point in the ordered spectrum where the tail becomes sufficiently spread out that noise can be dispersed across roughly $T$ weak directions. It separates the ``head'' (strong directions) from the ``tail'' (weak, noise-dispersing directions).

\begin{condition}[BLLT asymptotic conditions for benign overfitting]\label{cond:bllt}
    As $T \to \infty$ with covariance $\bm{\Sigma}_T$, the minimum-norm interpolator achieves $\mathrm{MSFE} \to \sigma^2$ if:
    \begin{equation}\label{eq:bllt_sufficient}
        \frac{r_0(\bm{\Sigma}_T)}{T} \to 0, \qquad \frac{k_T^*}{T} \to 0, \qquad \frac{T}{R_{k_T^*}(\bm{\Sigma}_T)} \to 0,
    \end{equation}
    and only if the slightly stronger conditions
    \begin{equation}\label{eq:bllt_necessary}
        \frac{r_0(\bm{\Sigma}_T)\log(1 + r_0(\bm{\Sigma}_T))}{T} \to 0, \qquad \frac{k_T^*}{T} \to 0, \qquad \frac{T}{R_{k_T^*}(\bm{\Sigma}_T)} \to 0
    \end{equation}
    hold.
\end{condition}

The three ratios capture:
\begin{enumerate}
    \item[(i)] \textbf{Overall effective rank} ($r_0/T \to 0$): the effective rank of the full spectrum, relative to sample size, must be controlled. Since $r_0 = tr(\bm{\Sigma})/\lambda_1$, this bounds the ratio of total variance to the leading eigenvalue.
    \item[(ii)] \textbf{Low-dimensional head} ($k^*/T \to 0$): only a vanishing fraction of the $T$ sample directions are needed to capture the strong spectral components.
    \item[(iii)] \textbf{Flat tail} ($T/R_{k^*} \to 0$): beyond the split index, the tail eigenvalues must be sufficiently uniform so that $R_{k^*} \gg T$. This ensures interpolated noise is spread across many directions of comparable variance rather than concentrated on a few dominant ones.
\end{enumerate}

These quantities and conditions follow \citet{bartlett2020benign} (their Definition~3, Definition~4, and Theorem~1). Their determinant normalization is without loss of generality: since $r_k$, $R_k$, and $k^*$ depend only on eigenvalue ratios, rescaling $\bm{\Sigma}$ does not affect whether \eqref{eq:bllt_sufficient}--\eqref{eq:bllt_necessary} hold.

\subsection{Extension to dependent data}\label{sec:dependent}

The BLLT theory assumes that the rows of $\bm{X}$ are drawn independently from a sub-Gaussian distribution. In macroeconomic applications the rows of $\bm{X}$ are consecutive time-series observations and therefore serially dependent. \citet{nakakita2022benign} extend the theory to dependent (sub-Gaussian) data and show that the BLLT spectral conditions still apply, with an adjustment for temporal dependence.

The key idea is a decorrelation argument: a multivariate whitening transformation removes the autocorrelation structure and restores the i.i.d.\ setting to which BLLT applies directly. The results then translate back to the original scale, and the BLLT bounds apply directly to the untransformed data, with an inflated variance term that accounts for the additional uncertainty from temporal dependence. The only additional requirement is a coherence condition (Assumption~\ref{coherence}) ensuring that the temporal dependence structure is such that an appropriate decorrelation transformation exists. 

Let $\{\mathbf{x}_t\}_{t=1}^T$ be a stationary $N$-dimensional process with autocovariance function
\[
    \bm{\Gamma}_X(h) := \E[\mathbf{x}_t \mathbf{x}_{t-h}'] \in \mathbb{R}^{N \times N}, \qquad h = 0, 1, 2, \ldots
\]
Let $\bm{\Sigma} = \bm{U}\bm{\Lambda}_\Sigma \bm{U}'$ be the spectral decomposition of $\bm{\Sigma} = \bm{\Gamma}_X(0)$, with eigenvectors $\bm{u}_1, \ldots, \bm{u}_N$ and eigenvalues $\lambda_1 \geq \cdots \geq \lambda_N > 0$. Each eigenvector $\bm{u}_i$ defines a linear combination of the $N$ original series; the index $i = 1, \ldots, N$ thus runs over these $N$ directions, which in a factor model correspond to the $k$ common factors (large eigenvalues) and the $N-k$ idiosyncratic components (small eigenvalues). The \emph{projected series} $\tilde{\mathbf{x}}_{it} := \bm{u}_i' \mathbf{x}_t$ is the scalar time series along the $i$-th such direction, with autocovariance
\begin{equation}\label{eq:gamma_i}
    \gamma_i(h) := \mathrm{Cov}(\tilde{\mathbf{x}}_{it},\, \tilde{\mathbf{x}}_{i,t-h}) = \bm{u}_i'\,\bm{\Gamma}_X(h)\,\bm{u}_i.
\end{equation}
For each $i$, let $\bm{\Xi}_{i,T}$ denote the $T \times T$ autocovariance matrix of the projected series $(\tilde{\mathbf{x}}_{i1}, \ldots, \tilde{\mathbf{x}}_{iT})'$.
When $\gamma_i(h)$ does not depend on $i$, we can write $\bm{\Xi}_{i,T} = \bm{\Xi}_T$ for all $i$ and the covariance of $\mathrm{vec}(\bm{X})$ has the separable Kronecker form $\bm{\Sigma} \otimes \bm{\Xi}_T$. In general, $\gamma_i(h)$ varies across directions, and the autocovariance takes the non-separable form
\[
    \mathrm{Cov}\big(\mathrm{vec}(\bm{X})\big) = \sum_{i=1}^N \Xi_{i,T} \otimes (u_i u_i').
\]

\paragraph{Decorrelation.}
\citet{nakakita2022benign} model the design matrix through the factorization
\[
    \bm{X} = \big[\bm{\Xi}_{1,T}^{1/2}\,z_1,\;\ldots,\;\bm{\Xi}_{N,T}^{1/2}\,z_N\big]\,\bm{\Lambda}_\Sigma^{1/2}\,\bm{U}^\top,
\]
where $\bm{z}_i \sim \mathcal{N}(0, \bm{I}_T)$ are i.i.d.\ In the separable case $\bm{\Xi}_{i,T} = \bm{\Xi}_T$, this reduces to $\bm{X} = \bm{\Xi}_T^{1/2}\,\bm{Z}\,\bm{\Lambda}_\Sigma^{1/2}\,\bm{U}^\top$, and pre-multiplying by $\bm{\Xi}_T^{-1/2}$ yields $\tilde{\bm{X}} = Z\Lambda_\Sigma^{1/2}U^\top$, which has i.i.d.\ rows with covariance $\bm{\Sigma}$, exactly the BLLT setting. The minimum-norm interpolant of the original problem coincides with that of the whitened problem, so the BLLT spectral conditions apply directly after decorrelation.
In the non-separable case, each projected series requires its own whitening transformation $\bm{\Xi}_{i,T}^{-1/2}$, and no single transformation decorrelates all directions simultaneously. The following assumption ensures these direction-specific structures are sufficiently similar that a single reference transformation approximately decorrelates all of them. 

\begin{assumption}\label{coherence}
{Coherent autocovariance.} The autocovariance matrices $\{\bm{\Xi}_{i,T}\}$ are mutually coherent: there exists a reference matrix $\bm{\Xi}_{0,T}$ and a constant $\epsilon \in (0,1]$ such that
\[
\epsilon \;\leq\; \lambda_{\min}(\bm{\Xi}_{0,T}^{-1}\bm{\Xi}_{i,T}) \;\leq\; 
\lambda_{\max}(\bm{\Xi}_{0,T}^{-1}\bm{\Xi}_{i,T}) \;\leq\; 1/\epsilon
\qquad \text{for all } i = 1, \ldots, N,
\]
where $\lambda_{\max}(\cdot)$ and $\lambda_{\min}(\cdot)$ denote the largest and smallest eigenvalues.
\end{assumption}

Coherence bounds how much $\gamma_i(h)$ (defined in~\eqref{eq:gamma_i}) can vary across eigendirections: it rules out situations where some directions are near-integrated while others are white noise. This condition is satisfied by any stationary VAR with companion matrix $\bm A$ and spectral radius $\rho(\bm A)<1-\delta$ for some $\delta>0$ and eigenvalues bounded away from zero, conditions met by standard macroeconomic VARs after applying the McCracken--Ng stationarity transformations. More broadly, the condition holds for any stationary VARMA process with bounded spectral density.

Crucially, \citet{nakakita2022benign} show that when coherence is satisfied, there is no need to actually perform any transformation: the usual BLLT bounds hold directly on the untransformed data, although with an inflated variance term on account of the temporal dependence.

\section{Factor Structure and Benign Overfitting}\label{sec:factor_bartlett}

This Section specializes the general setup of Section~\ref{sec:setup} and Section~\ref{sec:bllt_conditions} to factor-structured predictors. We derive the key spectral implications of both exact and approximate factor structures and verify when the BLLT conditions hold.

As detailed in Section~\ref{sec:bllt_conditions}, benign overfitting requires three conditions: (i) the overall effective rank must be small relative to the sample size; (ii) the signal must be concentrated in a low-dimensional subspace; (iii) the residual spectrum must be sufficiently flat to disperse noise. This Section shows that with factor-structured predictors, the first two conditions hold automatically as consequences of the factor structure, while the third condition, the tail condition, depends on the dispersion of the idiosyncratic variances. The tail flatness condition is easier to satisfy under exact factor models than under approximate factor models. In the exact case, the tail eigenvalues coincide with the idiosyncratic variances, so flatness requires only that no series has idiosyncratic variance orders of magnitude larger than others. In the approximate case, residual correlations among idiosyncratic components can concentrate variance into a few large tail eigenvalues, creating spikes that reduce the effective dimension of the tail and violate the flatness requirement. Readers not interested in the technical details can skip to Section~\ref{sec:augmentation_andRKHS}.

Our analysis is related to but distinct from \citet{bunea2021interpolating}, who study benign overfitting in a factor regression model, where the response variable depends directly on the latent factors. By contrast, we study linear regression where the \emph{predictors} have a factor structure but the response may depend on all $N$ predictors. This distinction matters for the bias condition: in a factor regression the signal is confined to the factor space by assumption, whereas in our setting the signal alignment with the factor space is an empirical question. Moreover, \citet{bunea2021interpolating} is based on an exact factor model within an i.i.d. setting, while we allow for an approximate factor structure and temporal dependence, both of which are typical features of macroeconomic data. 

\subsection{Spectral shape of data with a factor structure}
Recall the linear regression model shown in equation~\eqref{eq:model_stacked}:
\[
    \mathbf{y} = \mathbf{X}\boldsymbol{\beta}^* + \boldsymbol{\varepsilon}.
\]
Our goal is to estimate this model directly, without first extracting factors or reducing dimensionality. That is, the factor structure of the regressors is not used in the estimation step; rather, it provides the spectral configuration that enables benign overfitting. Specifically, we assume each row of $\mathbf{X}$ has a factor structure: $\mathbf{x}_t = \bm{\Lambda} \mathbf{f}_t + \mathbf{e}_t, t = 1, \ldots, T,$
where $\mathbf{f}_t \in \R^k$ is a vector of $k$ latent common factors, $\bm{\Lambda} \in \R^{N \times k}$ is the matrix of factor loadings, and $\mathbf{e}_t \in \R^N$ is the idiosyncratic error. The population covariance is
\[
    \bm{\Sigma} = \bm{\Lambda}\bm{\Lambda}^\prime + \bm{\Psi},
\]
where $\bm{\Psi} = \E[\mathbf{e}_t \mathbf{e}_t^\prime]$ is the idiosyncratic covariance matrix. The distinction between exact and approximate factor models revolves around the structure of $\bm{\Psi}$: the exact model assumes idiosyncratic errors are cross-sectionally uncorrelated ($\bm{\Psi}$ diagonal), while the approximate model permits cross-sectional correlations ($\bm{\Psi}$ non-diagonal but with bounded eigenvalues).
\begin{assumption}[\bf{Exact factor structure}]\label{ass:factor}
    \begin{enumerate}
        \item[(i)] $\E[\mathbf{f}_t] = \bm{0}$ and $\E[\mathbf{f}_t \mathbf{f}_t'] = \bm{I}_k$.
        \item[(ii)] $\E[\mathbf{e}_t] = \bm{0}$ and $\E[\mathbf{e}_t \mathbf{e}_t'] = \bm{\Psi} = \diag(\psi_1^2, \ldots, \psi_N^2)$ with $0 < \underline{\psi}^2 \leq \psi_j^2 \leq \bar{\psi}^2 < \infty$ for all $j$.
        \item[(iii)] $\mathbf{f}_t$ and $\mathbf{e}_s$ are mutually independent for all $t, s$.
        \item[(iv)] $\bm{\Lambda}'\bm{\Lambda} / N \to \bm{\Sigma}_\Lambda$ as $N \to \infty$, where $\bm{\Sigma}_\Lambda$ is $k \times k$ positive definite.
        \item[(v)] $N^{-1} \tr(\bm{\Psi}) \to \bar{\psi}_0^2 \in (0, \infty)$ as $N \to \infty$.
    \end{enumerate}
\end{assumption}
The diagonal structure of $\bm{\Psi}$ in part (ii) defines the exact factor model: idiosyncratic errors are cross-sectionally uncorrelated. This can be relaxed:
\begin{assumption}[\bf{Approximate factor structure}]\label{ass:factor_approx}
    Assumption~\ref{ass:factor} holds with (ii) replaced by:
    \begin{enumerate}
        \item[(ii$'$)] $\E[\mathbf{e}_t] = \bm{0}$ and $\E[\mathbf{e}_t \mathbf{e}_t'] = \bm{\Psi}$, where $\bm{\Psi}$ is positive definite with eigenvalues uniformly bounded: $0 < \underline{\psi}^2 \leq \lambda_{\min}(\bm{\Psi}) \leq \lambda_{\max}(\bm{\Psi}) \leq \bar{\psi}^2 < \infty$ for all $N$.
    \end{enumerate}
\end{assumption}
The approximate factor structure, introduced by \citet{ChamberlainRothschild1983}, relaxes the exact model by allowing weak cross-sectional correlation in the idiosyncratic errors. Under either assumption, the eigenvalues of $\bm{\Sigma}$ consist of $k$ \emph{spike} eigenvalues driven by the factors and $N - k$ \emph{tail} eigenvalues driven by the idiosyncratic component. Weyl's inequality gives
\begin{equation}\label{eq:weyl}
    \psi_{\min}^2 + \lambda_j(\bm{\Lambda}\bm{\Lambda}') \;\leq\; \lambda_j(\bm{\Sigma}) \;\leq\; \psi_{\max}^2 + \lambda_j(\bm{\Lambda}\bm{\Lambda}'), \qquad j = 1, \ldots, N,
\end{equation}
where $\lambda_j(\bm{\Lambda}\bm{\Lambda}') = 0$ for $j > k$, and $\psi_{\min}^2 := \lambda_{\min}(\bm{\Psi})$, $\psi_{\max}^2 := \lambda_{\max}(\bm{\Psi})$ denote the smallest and largest eigenvalues of the idiosyncratic covariance matrix $\bm{\Psi}$ (these satisfy $\underline{\psi}^2 \leq \psi_{\min}^2 \leq \psi_{\max}^2 \leq \bar{\psi}^2$, where $\underline{\psi}^2$ and $\bar{\psi}^2$ are the assumption-level bounds in Assumptions~\ref{ass:factor}--\ref{ass:factor_approx}).
The key spectral implication of a factor structure is a sharp separation between factor-driven and idiosyncratic eigenvalues, formalized in the following lemma.
\begin{lemma}[Spiked spectrum under factor structure]\label{lemma:spiked_spectrum}
Under either Assumption~\ref{ass:factor} or Assumption~\ref{ass:factor_approx}, the covariance $\bm{\Sigma} = \bm{\Lambda}\bm{\Lambda}' + \bm{\Psi}$ has a spiked eigenvalue structure:
\begin{enumerate}
    \item[(a)] The $k$ leading eigenvalues satisfy $\lambda_j(\bm{\Sigma}) = \lambda_j(\bm{\Lambda}\bm{\Lambda}') + O(1) = O(N)$ for $j = 1, \ldots, k$.
    \item[(b)] The remaining $N - k$ eigenvalues satisfy $\psi_{\min}^2 \leq \lambda_j(\bm{\Sigma}) \leq \psi_{\max}^2 = O(1)$ for $j = k+1, \ldots, N$.
\end{enumerate}
\end{lemma}
The proof is given in Appendix~\ref{proof:lemma_spiked_spectrum}. To verify the BLLT conditions, we define the tail moments
\[
    \bar{\lambda}_{\text{tail}} := \frac{1}{N-k}\sum_{j=k+1}^{N} \lambda_j, \qquad
    \overline{\lambda}^2_{\text{tail}} := \frac{1}{N-k}\sum_{j=k+1}^{N} \lambda_j^2,
\]
and the tail concentration ratio
\begin{equation}\label{eq:concentration}
    \mathcal{C} := \frac{\overline{\lambda}^2_{\text{tail}}}{(\bar{\lambda}_{\text{tail}})^2} \;\geq\; 1,
\end{equation}
with equality if and only if all tail eigenvalues are identical. The concentration ratio $\mathcal{C}$ measures the dispersion of the tail eigenvalues: $\mathcal{C} = 1$ indicates perfect tail flatness (all tail eigenvalues equal), while $\mathcal{C} \to \infty$ indicates tail concentration on a few dominant directions. 

The effective ranks satisfy\footnote{From the definitions \eqref{eq:effective_ranks}: $r_k = \sum_{i>k}\lambda_i / \lambda_{k+1} = (N-k)\bar{\lambda}_{\text{tail}} / \lambda_{k+1}$. For $R_k$: $R_k = (\sum_{i>k}\lambda_i)^2 / \sum_{i>k}\lambda_i^2 = ((N-k)\bar{\lambda}_{\text{tail}})^2 / ((N-k)\overline{\lambda}^2_{\text{tail}}) = (N-k)(\bar{\lambda}_{\text{tail}})^2 / \overline{\lambda}^2_{\text{tail}} = (N-k)/\mathcal{C}$.}
\begin{align}
    r_k(\bm{\Sigma}) &= \frac{(N-k)\,\bar{\lambda}_{\text{tail}}}{\lambda_{k+1}}, \qquad
    R_k(\bm{\Sigma}) = \frac{N-k}{\mathcal{C}}. \label{eq:Rk_general}
\end{align}
These objects are the key inputs for verifying the BLLT ratios under exact and approximate factor structures.

\subsection{BLLT conditions under a factor structure}
We now verify the BLLT conditions under the exact and approximate factor structures.

\begin{theorem}[BLLT conditions under factor structure]\label{prop:factor}
    Under Assumption~\ref{coherence} and either Assumption~\ref{ass:factor} (exact factor structure) or Assumption~\ref{ass:factor_approx} (approximate factor structure) with $N/T \to \infty$, the three BLLT conditions \eqref{eq:bllt_sufficient} are satisfied as follows:
    \begin{enumerate}
        \item[(i)] \emph{Overall effective rank.} The factor structure forces $r_0(\bm{\Sigma}) = O(1)$, so the first BLLT ratio vanishes: $r_0/T \to 0$.
        \item[(ii)] \emph{Low-dimensional head.} The spectral gap between the $k$ spike eigenvalues ($O(N)$) and the tail ($O(1)$) implies $k^* = k$. Since $k$ is fixed, the second BLLT ratio vanishes: $k^*/T \to 0$.
        \item[(iii)] \emph{Tail flatness.} The third BLLT condition requires
        \begin{equation}\label{eq:flatness_factor}
            \frac{T}{R_{k^*}(\bm{\Sigma})} = \frac{T\,\mathcal{C}}{N - k} \;\to\; 0.
        \end{equation}
        This condition is satisfied under any of the following sufficient conditions on the idiosyncratic covariance structure:
        \begin{enumerate}
            \item[(a)] \emph{Homoscedastic:} Under the exact factor model, $\bm{\Psi} = \psi^2 I_N$; under the approximate factor model, all eigenvalues of $\bm{\Psi}$ are equal to $\psi^2$. Then $\mathcal{C} = 1$ and $R_k = N - k$, so \eqref{eq:flatness_factor} holds since $T/(N-k) \to 0$.
            \item[(b)] \emph{Bounded heterogeneity:} $\bar{\psi}^2/\underline{\psi}^2 = O(1)$. Under the exact factor model, this bounds the ratio of diagonal entries of $\bm{\Psi}$; under the approximate factor model, it bounds the ratio of eigenvalues of $\bm{\Psi}$. Then $\mathcal{C} = O(1)$, so \eqref{eq:flatness_factor} holds since $T/(N-k) \to 0$.
            \item[(c)] \emph{Diverging concentration:} $\mathcal{C} \to \infty$ with $\mathcal{C} = o(N/T)$. Then \eqref{eq:flatness_factor} holds since $T\mathcal{C}/(N-k) = o(T \cdot N/T / N) = o(1)$.
        \end{enumerate}
    \end{enumerate}
    In all three cases, the BLLT conditions hold and benign overfitting occurs: the minimum-norm interpolator achieves MSFE converging to $\sigma^2$. The rate of convergence depends on $\mathcal{C}$: it is fastest in case~(a) and slowest in case~(c).
\end{theorem}

The proof is given in Appendix~\ref{proof:prop_factor}.

Conditions (i) and (ii) are satisfied by construction in a factor model, by virtue of Lemma~\ref{lemma:spiked_spectrum}. The key to condition (ii) is that $k^* = k$: the factor structure induces a sharp spectral jump from $O(N)$ eigenvalues (the factor spikes) to $O(1)$ eigenvalues (the idiosyncratic tail) at position $k$, so the natural split index is exactly the number of factors. In finite samples this gap can be weak but asymptotically the separation diverges under both factor model assumptions. 

For condition (iii), the critical requirement is that the idiosyncratic variances are not too dispersed across series. Under the exact factor model, where $\bm{\Psi}$ is diagonal, the concentration ratio $\mathcal{C}$ is determined directly by the dispersion of the diagonal entries $\psi_j^2$. If $\bar\psi^2/\underline\psi^2 = O(1)$, then $\mathcal{C} = O(1)$ and the flatness condition holds automatically as $N/T \to \infty$. Homoscedasticity ($\bm{\Psi} = \psi^2 I_N$) is the cleanest special case, giving $\mathcal{C} = 1$ and $R_k = N - k$. 

Under the approximate factor model, $\bm{\Psi}$ can have off-diagonal entries reflecting residual cross-sectional correlations, and $\mathcal{C}$ is determined by the eigenvalues $\mu_1, \ldots, \mu_N$ of $\bm{\Psi}$ rather than its diagonal entries (so that $\psi_{\min}^2 = \min_j \mu_j$ and $\psi_{\max}^2 = \max_j \mu_j$). The condition $\bar{\psi}^2/\underline{\psi}^2 = O(1)$ now bounds the ratio of the largest to smallest eigenvalues of $\bm{\Psi}$. Residual correlations can inflate $\mathcal{C}$ relative to the exact factor case: if idiosyncratic components are strongly positively correlated, their joint variance concentrates into fewer directions, making the tail less flat and benign overfitting less likely. In the extreme case where all idiosyncratic components are perfectly correlated, $\bm{\Psi}$ has rank one and $\mathcal{C} \to \infty$, violating the flatness condition entirely. In practice, if the factor structure captures most of the cross-sectional dependence, the residuals will be approximately uncorrelated and $\mathcal{C}$ will be close to its exact-factor-structure value.

\subsection{A tale of tails}

In Theorem~\ref{prop:factor}, the first two BLLT conditions from Condition~\ref{cond:bllt} — overall effective rank (i) and low-dimensional head (ii) — are automatically satisfied whenever the data-generating process follows a factor structure, an assumption that is reasonable in macroeconomic panels. The only empirically critical condition is the third BLLT condition (iii), the tail-flatness requirement $T/R_{k^*} \to 0$, which depends on the concentration ratio~\eqref{eq:concentration}. 

The concentration ratio $\mathcal{C}$ governs how effectively the spectral tail disperses interpolation noise. When tail eigenvalues are uniform ($\mathcal{C}=1$), noise is spread across many comparable directions and averages out at rate $T/(N-k)$. When the tail is heterogeneous ($\mathcal{C}>1$), noise concentrates in a few dominant directions and averaging is weaker. Equivalently, $R_{k^*}=(N-k)/\mathcal{C}$ is the number of \emph{effective} tail directions, and the flatness condition requires $T\mathcal{C}/(N-k)\to 0$.

    Which of cases~(a)--(c) in condition (iii) of Theorem~\ref{prop:factor} applies depends on how the spectrum of $\bm{\Psi}$ behaves as $N \to \infty$. As shown in the proof of Theorem~\ref{prop:factor}, the interlacing bounds imply that the tail eigenvalues of $\bm{\Sigma}$ are asymptotically governed by the ordered eigenvalues $\{\mu_j\}_{j=1}^N$ of $\bm{\Psi}$, so the concentration ratio is asymptotically determined by the cross-sectional moments of the $\mu_j$:
    \[
        \mathcal{C} \;\to\; \frac{\frac{1}{N-k}\sum_{j>k} \mu_j^2}{\Big(\frac{1}{N-k}\sum_{j>k} \mu_j\Big)^2} \qquad \text{as } N \to \infty.
    \]
    In case (a) we have $\mathcal{C} = 1$ and the tail is perfectly flat. We are in case~(b) if and only if this limit remains bounded, and in case~(c) if it diverges but slowly enough that $\mathcal{C} = o(N/T)$. If none of these cases apply---that is, if $\mathcal{C}$ grows faster than $N/T$---then condition~(iii) fails and benign overfitting does not obtain.

     Appendix~\ref{app:pareto_tail}  shows that a sufficient condition for case~(b) is that the eigenvalues $\{\mu_j\}$ have a Pareto tail with index $\alpha > 2$.
    This ensures $\mathcal{C} = O(1)$. If $1 < \alpha \leq 2$, $\mathcal{C}$ diverges but slowly enough to remain $o(N/T)$, placing us in case~(c). The case $\alpha \leq 1$ violates the bounded eigenvalue assumptions of the factor model and is ruled out. 
    

We apply these theoretical diagnostics to the FRED-MD and FRED-QD panels in Appendix~\ref{app:tail_diagnostics}, estimating the tail index $\alpha$ using Hill \citep{hill1975simple} and log-rank \citep{gabaix2011rank} methods at various split indices $k^*$. The finite-sample evidence is suggestive but inconclusive: at moderate splits ($k^* = 15$--$30$), tail-index estimates enter the range $1 < \alpha \leq 2$ consistent with case~(c), but the flatness ratio $T/R_{k^*}$ remains far from the asymptotic regime required for benign overfitting. In other words, with the available sample sizes ($N \sim 100$--$240$, $T \sim 260$--$800$), there is not enough data to reach the benign overfitting regime directly. This motivates the feature-expansion analysis in Section~\ref{sec:augmentation_andRKHS}, which achieves $N \gg T$ by construction and provides a direct implementation of the factor kernel without relying on asymptotic approximations.

\section{Feature Expansion and RKHS Regression}\label{sec:augmentation_andRKHS}

Section~\ref{sec:factor_bartlett} showed that, when the data-generating process follows a factor structure, benign-overfitting conditions are satisfied in high dimensions, but in standard macro panels $N$ and $T$ are of similar order, so the condition $N\gg T$ must be engineered rather than assumed. This section addresses the issue in two steps. First, Section~\ref{sec:feature_expansion} introduces a synthetic data augmentation procedure that artificially expands the feature dimension to achieve $N \gg T$, thereby enabling benign overfitting in finite samples. Second, Section~\ref{sec:rkhs} shows that as the number of synthetic replicates grows, this augmentation converges to a closed-form kernel estimator in a Reproducing Kernel Hilbert Space (RKHS). This limiting representation is valuable because it delivers the benefits of benign overfitting without requiring practitioners to generate or store synthetic data. More broadly, the kernel duality reveals that benign overfitting, when it works, succeeds because overparameterization implicitly constructs a well-behaved kernel, not because overparameterization is intrinsically desirable.

\subsection{Feature expansion via synthetic data}\label{sec:feature_expansion}

The procedure operates in two steps. First, we estimate the factor model parameters from the observed data. Second, we use the estimated model to generate $B$ synthetic datasets by drawing fresh idiosyncratic noise while holding the estimated factor path fixed, thereby expanding the cross-sectional dimension without altering the time-series structure.

\medskip

\noindent Recall from Section~\ref{sec:factor_bartlett} the factor model
\begin{equation}\label{eq:factor_model_sec5}
    \mathbf{x}_t = \boldsymbol{\Lambda}\mathbf{f}_t + \mathbf{e}_t, \qquad t = 1, \ldots, T,
\end{equation}
where $\mathbf{f}_t \in \R^k$ contains the $k$ common factors, $\boldsymbol{\Lambda} \in \R^{N \times k}$ is the loadings matrix, and $\mathbf{e}_t \in \R^N$ is the idiosyncratic error with $\E[\mathbf{e}_t \mathbf{e}_t'] = \boldsymbol{\Psi}$. We estimate $\hat{\boldsymbol{\Lambda}}$ and the factor path $\hat{\mathbf{f}}_1,\ldots,\hat{\mathbf{f}}_T$ by principal components, with the number of factors $k$ selected via the information criterion of \citet{bai2002determining}. Given the estimated factors, the idiosyncratic covariance is estimated from the residuals $\hat{\mathbf{e}}_t = \mathbf{x}_t - \hat{\boldsymbol{\Lambda}}\hat{\mathbf{f}}_t$:
\begin{equation}\label{eq:psi_hat}
    \hat{\psi}_i^2 \;=\; \frac{1}{T}\sum_{t=1}^{T}\hat{e}_{it}^2, \qquad i = 1,\ldots,N,
\end{equation}
giving $\hat{\boldsymbol{\Psi}} = \mathrm{diag}(\hat{\psi}_1^2,\ldots,\hat{\psi}_N^2)$. 

\medskip

\noindent Given $\hat{\boldsymbol{\Lambda}}$, $\hat{\boldsymbol{\Psi}}$, and the estimated factor path $\hat{\mathbf{f}}_1,\ldots,\hat{\mathbf{f}}_T$, we generate $B$ synthetic datasets by keeping the common component fixed and resampling only the idiosyncratic noise:
\[
    \mathbf{x}_t^{*(b)} = \hat{\boldsymbol{\Lambda}}\hat{\mathbf{f}}_t +
    \mathbf{e}_t^{*(b)}, \qquad
    \mathbf{e}_t^{*(b)} \sim \mathcal{N}(\mathbf{0}, \hat{\boldsymbol{\Psi}}),
\]
for $b = 1, \ldots, B$ and $t = 1, \ldots, T$. The common component $\hat{\boldsymbol{\Lambda}}\hat{\mathbf{f}}_t$ is identical across replications and carries all temporal dependence through the estimated factor path. Only the idiosyncratic noise is redrawn, independently across $t$ and $b$. We have also experimented with a bootstrap approach to resampling the errors, which allows for any correlation pattern across the residuals, and found the results to be very similar.
 
Each synthetic dataset $\mathbf{X}^{*(b)} = [\mathbf{x}_1^{*(b)}, \ldots, \mathbf{x}_T^{*(b)}]' \in \R^{T \times N}$ has the same dimensions as the original data matrix $\mathbf{X} = [\mathbf{x}_1, \ldots, \mathbf{x}_T]' \in \R^{T \times N}$. The augmented predictor matrix stacks the original and all $B$ synthetic datasets horizontally:
\begin{equation}\label{eq:synth_aug}
    \bZ = \Big[\mathbf{X} \;\Big|\; \mathbf{X}^{*(1)} \;\Big|\;
    \cdots \;\Big|\; \mathbf{X}^{*(B)}\Big]
    \in \R^{T \times (B+1)N}.
\end{equation}
This expands the feature dimension from $N$ to $(B+1)N$ while keeping the sample size fixed at $T$. For sufficiently large $B$, the ratio $(B+1)N/T$ can be made arbitrarily large, placing the regression problem firmly in the overparameterized regime where $N \gg T$.

\subsection{RKHS representation}\label{sec:rkhs}

A Reproducing Kernel Hilbert Space (RKHS) is a Hilbert space of functions in which evaluation at any point is a continuous linear functional. Concretely, an RKHS is defined by a positive definite kernel $k(\cdot, \cdot)$: for any set of observations $\mathbf{z}_1, \ldots, \mathbf{z}_T$, the kernel matrix $\bm{K}$ with entries $K_{ts} = k(\mathbf{z}_t, \mathbf{z}_s)$ encodes pairwise similarities, and the minimum-norm interpolator in the RKHS predicts via $\hat{y}_{\mathrm{new}} = \mathbf{k}'(\mathbf{z}_{\mathrm{new}}) \bm{K}^{-1} \mathbf{y}$, where $\mathbf{k}(\mathbf{z}_{\mathrm{new}})$ collects the kernel evaluations between the new point and each training point. The key insight of this subsection is that the minimum-norm interpolator applied to the augmented system of Section~\ref{sec:feature_expansion} is equivalent to kernel regression with a specific factor-structured kernel, and that this equivalence becomes exact as $B \to \infty$.

The most direct approach to exploiting the augmented matrix $\bZ$ in~\eqref{eq:synth_aug} is to estimate the regression model
\[
    y_t = \mathbf{z}_t'\boldsymbol{\beta} + \varepsilon_t, \qquad t = 1, \ldots, T,
\]
where $\mathbf{z}_t \in \R^{(B+1)N}$ is the $t$-th row of $\bZ$, stacking $\mathbf{x}_t$ with the $B$ synthetic replicates $\mathbf{x}_t^{*(1)}, \ldots, \mathbf{x}_t^{*(B)}$, using the minimum-norm interpolator:
\begin{equation}\label{eq:synth_estimator}
    \hat{\boldsymbol{\beta}} = \bZ'(\bZ\bZ')^{-1}\mathbf{y},
\end{equation}
which is the ridgeless estimator~\eqref{eq:ridgeless} applied to the $(B+1)N$-dimensional augmented system. However, this direct approach becomes computationally demanding as $B$ grows: $\hat{\boldsymbol{\beta}}$ lives in a space that expands with every new replicate. A key duality result shows that predictions from~\eqref{eq:synth_estimator} can be computed equivalently via kernel regression operating only in the $T$-dimensional observation space, without ever forming $\hat{\boldsymbol{\beta}}$ explicitly. This section establishes that equivalence and derives the limiting kernel as $B \to \infty$.

\begin{lemma}[Kernel duality of minimum-norm 
interpolation]\label{lem:kernel_duality}
    Let $\bZ \in \mathbb{R}^{T \times P}$ with
    $P > T$ and $\bZ\bZ'$ invertible. The
    minimum-norm interpolator
    $\hat{\boldsymbol{\beta}} = \bZ'(\bZ\bZ')^{-1}\mathbf{y}$ predicts
    a new observation $\mathbf{z}_{\textup{new}} \in
    \mathbb{R}^P$ as
    \[
        \hat{y}_{\textup{new}} \;=\;
        \mathbf{z}_{\textup{new}}'\hat{\boldsymbol{\beta}} \;=\;
        \mathbf{k}' (\bZ\bZ')^{-1} \mathbf{y},
    \]
    where $\mathbf{k} \in \mathbb{R}^T$ has entries
    $k_s = \mathbf{z}_{\textup{new}}'\mathbf{z}_s = \langle \mathbf{z}_{\textup{new}}, \mathbf{z}_s \rangle$, the inner product between the new observation and the $s$-th observation in the augmented feature space. That is, the
    prediction depends on $\bZ$ only through the
    Gram matrix $\bZ\bZ'$ and the kernel
    evaluations $k_s$ between the new observation and
    the training points.
\end{lemma}
\begin{proof}
    See Appendix~\ref{proof:kernel_duality}.
\end{proof}

Lemma~\ref{lem:kernel_duality} establishes the kernel duality: predictions can be computed without ever forming the $(B+1)N$-dimensional coefficient vector $\hat{\boldsymbol{\beta}}$, depending only on the $T \times T$ Gram matrix $\bZ\bZ'$ and kernel evaluations between the new observation and
the augmented data from the whole estimation sample. For the synthetic augmentation in~\eqref{eq:synth_aug}, the Gram matrix entries are
\[
    (\bZ\bZ')_{ts} = \mathbf{x}_t^\prime\mathbf{x}_s + \sum_{b=1}^B \mathbf{x}_t^{*(b)\prime} \mathbf{x}_s^{*(b)}.
\]
As $B \to \infty$, the sum of synthetic inner products converges to an expectation over the idiosyncratic noise distribution, yielding a closed-form kernel that depends only on the estimated factor structure. Proposition~\ref{thm:synth_kernel} characterizes this limiting kernel.

\begin{proposition}[Synthetic augmentation implements 
factor-structured kernel ridge regression]
\label{thm:synth_kernel}
    As $B \to \infty$, the minimum-norm interpolator~\eqref{eq:synth_estimator}
    on $\bZ$ converges to kernel ridge regression:
    \[
        \hat{y} \;\to\; \bm{K}_{\textup{synth}}
        (\bm{K}_{\textup{synth}} + \lambda \bm{I}_T)^{-1}\mathbf{y},
    \]
    where the kernel matrix has entries
    \begin{equation}\label{eq:synth_kernel_entries}
        K_{\textup{synth}}(t,s) \;=\;
        \hat{\mathbf{f}}_t'\hat{\boldsymbol{\Lambda}}'\hat{\boldsymbol{\Lambda}}
        \hat{\mathbf{f}}_s,
    \end{equation}
    and the ridge parameter is
    $\lambda = \textup{tr}(\hat{\boldsymbol{\Psi}})$. The temporal
    structure of $\bm{K}_{\textup{synth}}$ is inherited
    entirely from the estimated factor path
    $\hat{\mathbf{f}}_1, \ldots, \hat{\mathbf{f}}_T$, without
    assuming a parametric law of motion for the
    factors.
\end{proposition}
\begin{proof}
    See Appendix~\ref{proof:synth_kernel}.
\end{proof}

Proposition~\ref{thm:synth_kernel} reveals that benign overfitting in factor models works through the implicit construction of a factor-structured kernel. The factor structure ensures the kernel preserves the low-dimensional signal (the common component) while the tail flatness condition ensures noise is dispersed across many effective directions. The split index $k^*$ identifies the boundary between the signal subspace (which the kernel must preserve) and the noise-dispersing tail (which the kernel must spread evenly), while the tail flatness measure $R_{k^*}$ quantifies how effectively the kernel disperses noise. This explains why overparameterization succeeds when it does: not because adding parameters is intrinsically beneficial, but because it constructs a kernel that separates signal from noise in an appropriate way. 

The kernel~\eqref{eq:synth_kernel_entries} measures temporal proximity through the estimated factor structure rather than calendar time: periods $t$ and $s$ receive high weight when their factor realizations $\hat{\mathbf{f}}_t$ and $\hat{\mathbf{f}}_s$ are close in the metric induced by $\hat{\boldsymbol{\Lambda}}'\hat{\boldsymbol{\Lambda}}$. The ridge parameter $\lambda = \textup{tr}(\hat{\boldsymbol{\Psi}})$ provides automatic regularization scaled by idiosyncratic variance. This is different from the kernel estimators typically used in the econometrics literature, as those are local methods which borrow from adjacent dates, while this kernel estimator is a global method which borrows from economically similar dates, regardless of calendar distance. This property is valuable in regime-dependent environments where comparable macroeconomic conditions reappear non-consecutively.

Finally, since the augmentation strategy converges to kernel ridge regression with a closed-form kernel, practitioners can implement the kernel estimator directly without generating synthetic data. This avoids the computational cost and Monte Carlo variability of finite-$B$ augmentation. A detailed RKHS formulation and the connection to the forecasting implementation are provided in Appendix~\ref{app:rkhs_formulation}.

\section{Forecasting}\label{sec:empirical}

We now evaluate the factor kernel of Proposition~\ref{thm:synth_kernel} in a comprehensive rolling-window out-of-sample forecasting exercise covering all available series in the FRED-MD dataset. Unlike the motivating example of Section~\ref{sec:motivation}, which used a single holdout period to trace out the full double-descent curve, this will constitute a proper pseudo real-time forecasting evaluation: every forecast will be computed using only information available at the time it is made, the window will roll forward one observation at a time, and performance will be measured by the mean squared forecast error (MSFE) accumulated over the full evaluation sample.

\subsection{Data and evaluation design}

We use the February 2026 vintage of the FRED-MD monthly dataset \citep{mccracken2016fred}, which contains $N_{\mathrm{raw}} = 127$ monthly macroeconomic and financial series. All series are transformed using the \citet{mccracken2016fred} transformation codes (log differences, second differences, etc., as indicated by \texttt{tcode}). We end the sample on October 2025 to avoid real-time data revision issues that are most severe at the end of the sample.\footnote{Series with more than 20\% missing observations are excluded; remaining missing predictor values are imputed with column medians computed in the training window. After cleaning, a typical target uses about $N \approx 120$ predictors.}

\medskip

\noindent For each of the $N_{\mathrm{raw}}$ series and each forecast horizon $h \in \{1, 3, 6, 12\}$ months, we compute direct $h$-step-ahead forecasts. The general forecasting model is
\begin{equation}\label{eq:general_forecast}
    y_{t+h} = \phi\, y_{t} + f(\mathbf{x}_{t}) + \varepsilon_{t+h},
\end{equation}
where $y_t$ is the target variable available at forecast origin $t$, and $\mathbf{x}_t$ is the predictor vector observed at the same origin, obtained by pulling together all the remaining series in the panel. The AR(1) term $\phi\, y_t$ is common to all models; they differ only in the specification of $f(\cdot)$. Because the forecast is direct (we are forecasting $y_{t+h}$), the model is estimated separately for each horizon $h$. We use a rolling window estimation with $W = 120$ months. At forecast origin $t$, the model is estimated on observations $s=t-W-h+1,\ldots,t-h$. With data spanning January 1959 to October 2025, the first forecast origin is January 1969 for $h=1$ (forecasting February 1969), providing approximately 56 years of out-of-sample evaluation. The number of out-of-sample forecasts is $P=T-h-W$.

At each forecast origin, variables are standardized using moments from the current training window only: predictor means and standard deviations for $\mathbf{x}_s$, and separate mean and standard deviation for the dependent variable and its lag. Standardization is recomputed at each origin. Squared forecast errors are accumulated in standardized units, so MSFE is dimensionless and comparable across targets.

For each model $M$ and each target--horizon pair $(y, h)$, the out-of-sample MSFE is
\begin{equation}\label{eq:msfe_oos}
    \widehat{\mathrm{MSFE}}_{y,h}^{(M)} = \frac{1}{P} \sum_{p=1}^{P} \bigl(\hat{y}^{(M)}_{t_p+h\mid t_p} - y_{t_p+h}\bigr)^2,
\end{equation}
computed in standardized units. The relative performance of a model $M$ against a benchmark $B$ is summarized by the RMSE ratio
\begin{equation}\label{eq:imp_ratio}
    \mathrm{Imp}_{y,h} = \frac{\widehat{\mathrm{RMSE}}_{y,h}^{(M)}}{\widehat{\mathrm{RMSE}}_{y,h}^{(B)}} = \sqrt{\frac{\widehat{\mathrm{MSFE}}_{y,h}^{(M)}}{\widehat{\mathrm{MSFE}}_{y,h}^{(B)}}},
\end{equation}
where values below one indicate that model $M$ improves over the benchmark $B$. In our evaluation, $M = \mathrm{FK}$ and $B = \mathrm{FM}$.

\subsection{Models}\label{sec:emp_methods}

We compare two models. The first one, which will serve as our benchmark, is a factor model estimated via OLS (FM). This is a special case of~\eqref{eq:general_forecast} with $f(\mathbf{x}_t) = \boldsymbol{\gamma}_h'\hat{\mathbf{f}}_t$, a linear function of the factors extracted from the panel of predictors. This is the diffusion-index forecast of \citet{stock2002forecasting}, augmented with an AR(1) term. For each origin $t$ and horizon $h$, we estimate by OLS
\[
    y_{s+h} = \phi_h y_s + \boldsymbol{\gamma}_h'\hat{\mathbf{f}}_s + u_{s+h}, \qquad s=t-W-h+1,\ldots,t-h,
\]
where $\hat{\mathbf{f}}_s$ are the factors extracted from the rolling window using principal components. The forecast is
\[
    \hat{y}^{\mathrm{FM}}_{t+h\mid t} = \hat\phi_h y_t + \hat{\boldsymbol{\gamma}}_h'\hat{\mathbf{f}}_t.
\]
At each forecast origin, the number of factors is selected with the Bai--Ng $\mathrm{IC}_{p2}$ criterion \citep{bai2002determining} from the standardized predictor panel in the rolling window, with $k_{\max}=20$. This factor-selection step determines the number of factors used in both models.

\medskip

\noindent Our second model is a semiparametric factor-based kernel regression estimated via GLS (FK). This is also a special case of~\eqref{eq:general_forecast} with $f(\mathbf{x}_t)$ replaced by a nonparametric kernel component, the closed-form limit of the synthetic augmentation strategy from Proposition~\ref{thm:synth_kernel}. As derived below, this yields $f(\mathbf{x}_t) = \hat{\mathbf{k}}_{t,h}'\hat{\boldsymbol{\alpha}}_h$, where $\hat{\mathbf{k}}_{t,h}$ collects the kernel evaluations between origin $t$ and the training-window dates; the RKHS derivation is in Appendix~\ref{app:rkhs_formulation}.

For a fixed target, horizon $h$, and forecast origin $t$,  define
\[
    \mathbf{y}_h = (y_{t-W+1+h},\ldots,y_{t+h})', \qquad \mathbf{y}_{-1} = (y_{t-W+1},\ldots,y_t)'.
\]
Let $\hat{\mathbf{K}}_h$ be the $W\times W$ estimated factor kernel matrix over window dates, with entries $\hat{\mathbf{K}}_h(s,r)=\hat{\mathbf{f}}_s'\hat{\boldsymbol{\Lambda}}'\hat{\boldsymbol{\Lambda}}\hat{\mathbf{f}}_r$, and let
\[
    \hat{\mathbf{K}}_h^{(\lambda)} = \hat{\mathbf{K}}_h + \lambda \mathbf{I}_W,
\]
where $\lambda = \mathrm{tr}(\hat{\boldsymbol{\Psi}})$ is the ridge parameter determined by the trace of the estimated idiosyncratic covariance matrix, as established in Proposition~\ref{thm:synth_kernel}. The semiparametric factor kernel model with GLS AR(1) is
\[
    \mathbf{y}_h = \phi_h \mathbf{y}_{-1} + \hat{\mathbf{K}}_h\boldsymbol{\alpha}_h + \mathbf{u}_h.
\]
Using the Moore--Penrose inverse of $\hat{\mathbf{K}}_h^{(\lambda)}$, the GLS estimates are
\begin{equation}\label{eq:app_gls_beta}
    \hat{\phi}_h = \big(\mathbf{y}_{-1}'(\hat{\mathbf{K}}_h^{(\lambda)})^+\mathbf{y}_{-1}\big)^{-1}\mathbf{y}_{-1}'(\hat{\mathbf{K}}_h^{(\lambda)})^+\mathbf{y}_h,
\end{equation}
\begin{equation}\label{eq:app_gls_alpha}
    \hat{\boldsymbol{\alpha}}_h = (\hat{\mathbf{K}}_h^{(\lambda)})^+\big(\mathbf{y}_h - \hat\phi_h \mathbf{y}_{-1}\big).
\end{equation}
For forecast origin $t$, define the kernel vector to the window dates,
\[
    \hat{\mathbf{k}}_{t,h} = \big(\hat{\mathbf{K}}_h(t,t-W+1),\ldots,\hat{\mathbf{K}}_h(t,t)\big)'.
\]
Then the direct forecast is
\begin{equation}\label{eq:app_gls_forecast}
    \hat{y}^{\mathrm{FK}}_{t+h\mid t} = \hat\phi_h y_t + \hat{\mathbf{k}}_{t,h}'\hat{\boldsymbol{\alpha}}_h.
\end{equation}
Larger idiosyncratic dispersion (larger $\lambda$) implies stronger regularization automatically. The key difference from the benchmark is that the factor kernel measures temporal proximity through the estimated factor structure: periods $t$ and $s$ receive high weight when their factor realizations $\hat{\mathbf{f}}_t$ and $\hat{\mathbf{f}}_s$ are close, regardless of calendar distance. This state-dependent pooling is valuable in regime-dependent environments where comparable macroeconomic conditions reappear non-consecutively.

\subsection{Results}\label{sec:results}
Figure~\ref{fig:oos_hist} shows the cross-sectional distribution of $\mathrm{Imp}_{y,h}$ across all FRED-MD series, separately for each horizon. A distribution skewed to the left of one indicates that the factor kernel systematically outperforms the FM benchmark in the cross-section. The vertical dashed line marks the break-even point.

\begin{figure}[h!]
    \centering
    \includegraphics[width=\linewidth]{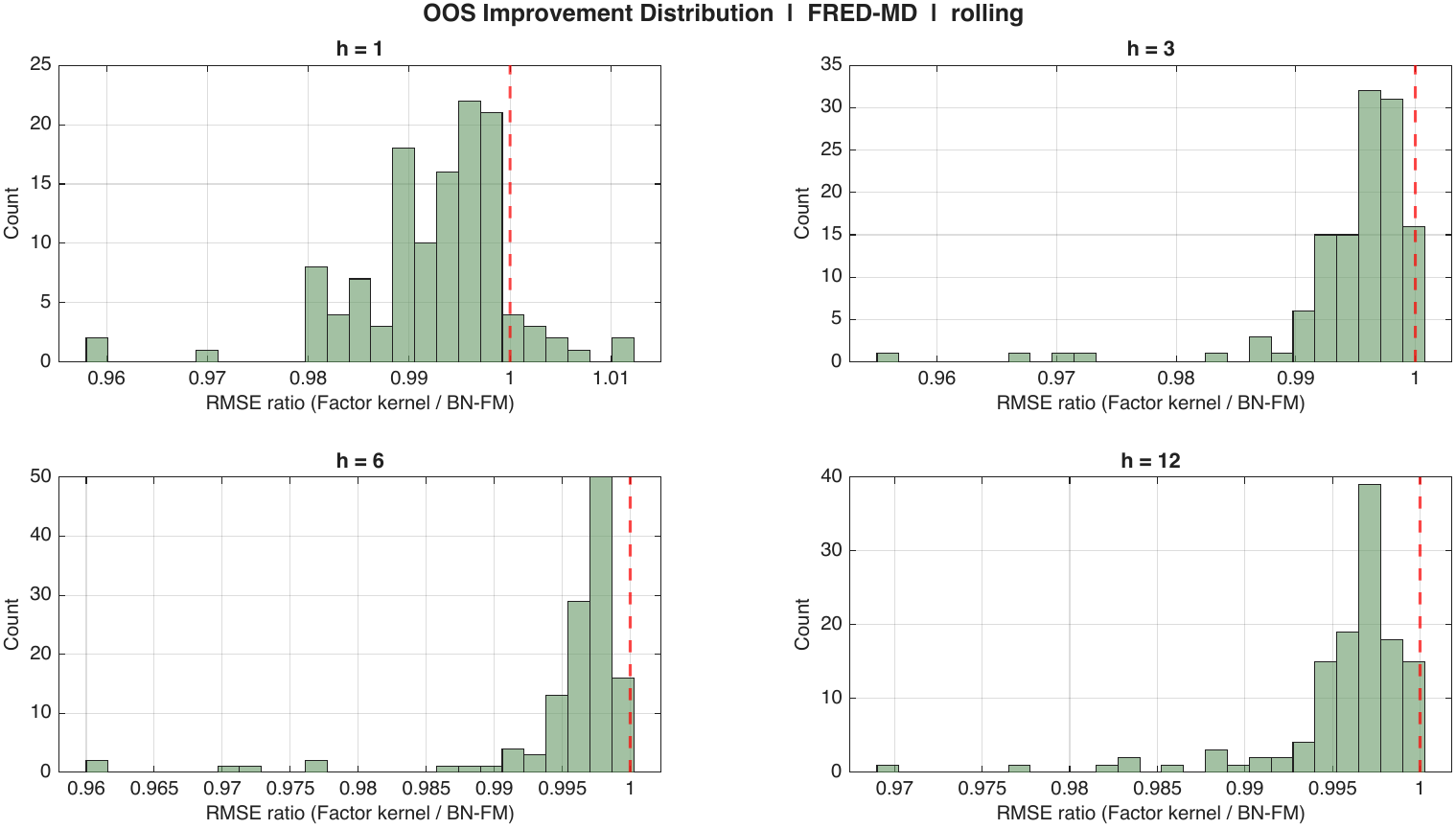}
    \caption{Distribution of the relative improvement of the factor kernel over the
             FM benchmark across all FRED-MD variables, by forecast horizon
             ($h = 1, 3, 6, 12$ months; one panel per horizon).
             The horizontal axis shows $\mathrm{RMSE}_{\mathrm{FK}} /
             \mathrm{RMSE}_{\mathrm{FM}}$; values to the left of the dashed vertical
             line (red) indicate that the factor kernel achieves a lower MSFE than the
             FM benchmark.
             MSFE is computed over the out-of-sample evaluation period using a rolling
             estimation window of $W = 120$ months,
             Bai--Ng factor selection, and rolling standardization.}
    \label{fig:oos_hist}
\end{figure}

\paragraph{Does the factor kernel consistently beat the FM benchmark?}
The FM model is a strong baseline: it uses the same factors and the same AR(1) term as FK, but combines them through linear OLS rather than kernel-weighted GLS. Any FK advantage therefore comes from the kernel metric and from state-dependent pooling (dates with similar factor realizations receive more weight), as discussed in Section~\ref{sec:rkhs} and formalized in equations~\eqref{eq:app_gls_beta}--\eqref{eq:app_gls_forecast}.

Figure~\ref{fig:oos_hist} shows that, across all FRED-MD variables, the distribution of $\mathrm{Imp}_{y,h}$ is centered to the left of one at every horizon considered: the factor kernel delivers a lower MSFE than the FM benchmark for 92\% of targets at $h=1$ and 99\% at $h=6$ and $h=12$. The mean RMSE ratio is 0.99 at $h=1$, indicating a 1\% mean improvement. While the magnitude of improvement is modest, it is highly consistent across targets and becomes increasingly pervasive at longer horizons. Statistical significance testing using the Diebold--Mariano test (two-sided, 5\% level) reveals that 34\% of FK wins at $h=1$ are statistically significant, rising to 54\% at $h=6$ and 50\% at $h=12$. The increase in both win rate and statistical significance at longer horizons indicates that FK's advantage becomes both more reliable and more pervasive as the forecast horizon lengthens.

\paragraph{Which variables benefit most from the factor kernel?}
To interpret and understand the heterogeneity in the performance of the factor kernel, we organize the FRED-MD variables in terms of their \textit{persistence} (lag-1 autocorrelation of the standardized variable in the second half of the sample) and map this against relative out-of-sample performance. Figure~\ref{fig:oos_ac1} shows a clear negative relationship between persistence and RMSE ratio, with the correlation strengthening from $-0.34$ at $h=1$ to $-0.54$ at $h=6$. Highly persistent variables benefit substantially more from the factor kernel's state-dependent pooling: the kernel identifies economically similar states across time, which is particularly valuable when the target exhibits strong temporal dependence. FK's advantage becomes increasingly pervasive (winning for nearly all targets) and statistically robust at longer horizons, even as the magnitude of improvement becomes more modest and uniform across targets. Persistence maintains its role as the primary predictor of FK advantage even at $h=12$.

\begin{figure}[t!]
    \centering
    \includegraphics[width=\linewidth]{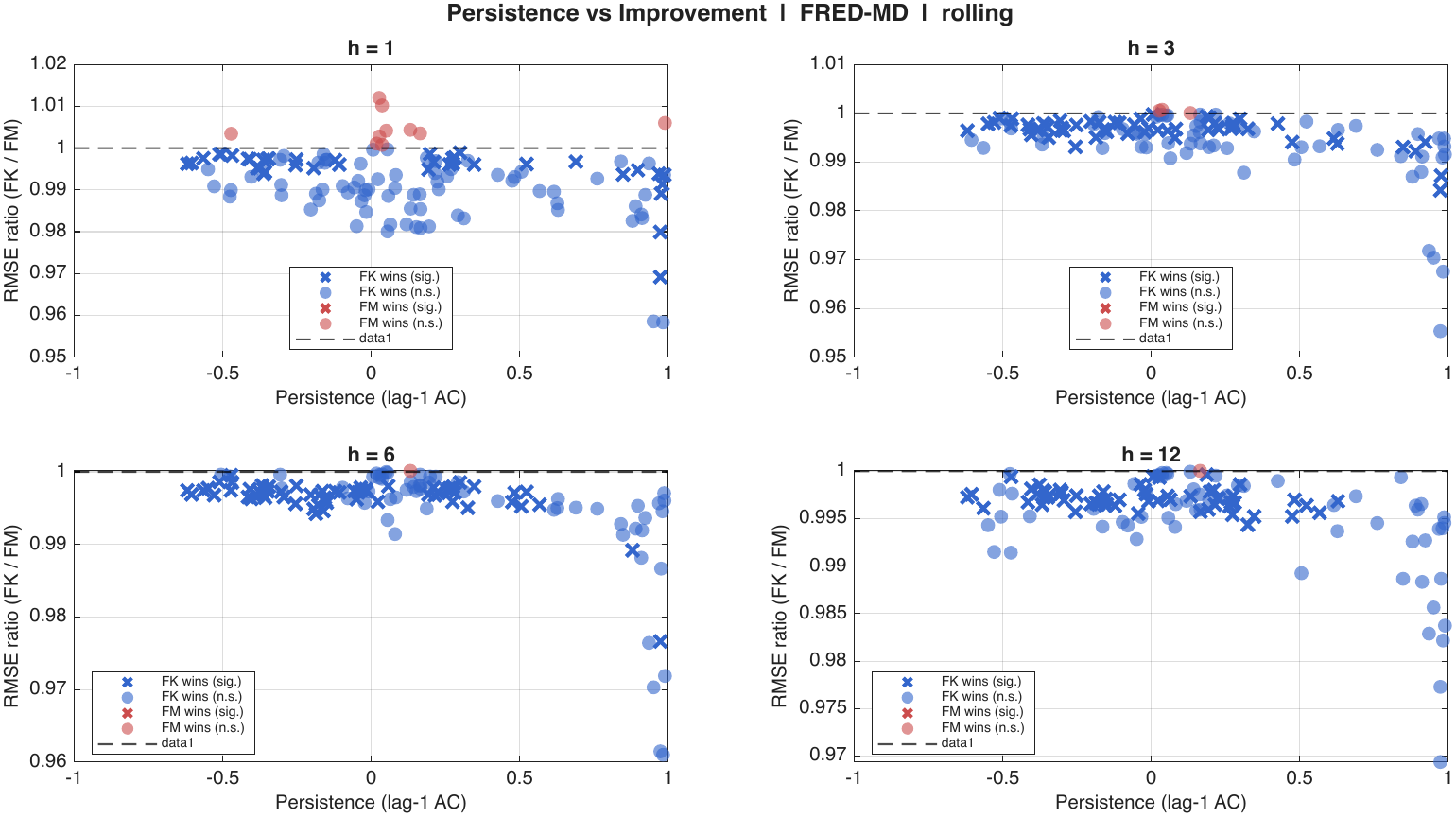}
    \caption{Variable persistence vs.\ relative improvement of the factor kernel
             over the FM benchmark, by forecast horizon ($h = 1, 3, 6, 12$ months).
             Each point represents one FRED-MD variable.
             The horizontal axis shows the lag-1 autocorrelation of the (standardized)
             target series, computed on the second half of the available sample.
             The vertical axis shows $\mathrm{RMSE}_{\mathrm{FK}} /
             \mathrm{RMSE}_{\mathrm{FM}}$.
             Blue markers: factor kernel wins; red markers: FM wins.
             Crosses indicate statistical significance at the 5\% level by the Diebold--Mariano test.
             The dashed horizontal line at 1 marks the break-even point.
             Variables are standardized using rolling training-window moments;
             MSFE is accumulated over the rolling evaluation window.}
    \label{fig:oos_ac1}
\end{figure}

\subsection{Quarterly data results}\label{sec:results_qd}

In this section, we repeat the forecasting evaluation exercise above using the FRED-QD quarterly dataset \citep{mccracken2016fred}. The quarterly data provide a complementary perspective on the factor kernel's performance in a setting where the cross-sectional dimension systematically exceeds the time-series dimension, placing the estimation problem squarely in the interpolation regime throughout the evaluation period.

\medskip 

\noindent The February 2026 vintage of FRED-QD contains $N_{\mathrm{raw}} = 248$ quarterly macroeconomic and financial series. After excluding series with more than 20\% missing observations and imputing remaining missing values with training-window medians, we are left with approximately $N \approx 240$ predictors. The sample spans 1960:Q1 to 2025:Q3 and the rolling window length is set to $W = 40$ quarters (10 years). Forecast horizons are $h \in \{1, 2, 4, 8\}$ quarters. The first forecast origin is 1969:Q1, providing approximately 56 years of out-of-sample evaluation. Because $N \approx 240 \gg W = 40$, every rolling window lies in the overparameterized regime $N > T$, so the ridgeless estimator fits the training data exactly at each forecast origin. This contrasts with the monthly data, where $N \approx 120 < W = 120$ in most windows, placing the estimation in the underparameterized regime.

The models, estimation procedures, and performance metrics are identical to those described in Sections~\ref{sec:emp_methods} and \ref{sec:results}. We rely on the Bai--Ng $\mathrm{IC}_{p2}$ criterion to select the number of factors at each origin with $k_{\max} = 20$. The factor kernel is constructed from the estimated factors and loadings as in Proposition~\ref{thm:synth_kernel}, and the ridge parameter $\lambda = \mathrm{tr}(\hat{\boldsymbol{\Psi}})$ is determined by the idiosyncratic variance. The relative performance measure $\mathrm{Imp}_{y,h} = \mathrm{RMSE}_{\mathrm{FK}} / \mathrm{RMSE}_{\mathrm{FM}}$ is computed for each target and horizon.

Figure~\ref{fig:oos_hist_QD} displays the cross-sectional distribution of $\mathrm{Imp}_{y,h}$ across all FRED-QD series, separately for each horizon. As in the monthly case, a distribution concentrated to the left of one indicates systematic outperformance by the factor kernel. The results show that the factor kernel delivers lower MSFE than the FM benchmark for 96\% of targets at $h=1$, $h=2$, and $h=8$, and 99\% at $h=4$, with the mean RMSE ratio of 0.95 at $h=1$ indicating a 5\% mean improvement. The mean improvement is most pronounced at short horizons and attenuates to approximately 3\% at longer horizons, consistent with the well-documented decline in factor-model predictability as the forecast horizon lengthens. Statistical significance testing reveals that 9\% of FK wins at $h=1$ are statistically significant by the Diebold--Mariano test, rising to 30\% at $h=8$, indicating that the improvements become more statistically robust at longer horizons.

\begin{figure}[h!]
    \centering
    \includegraphics[width=\linewidth]{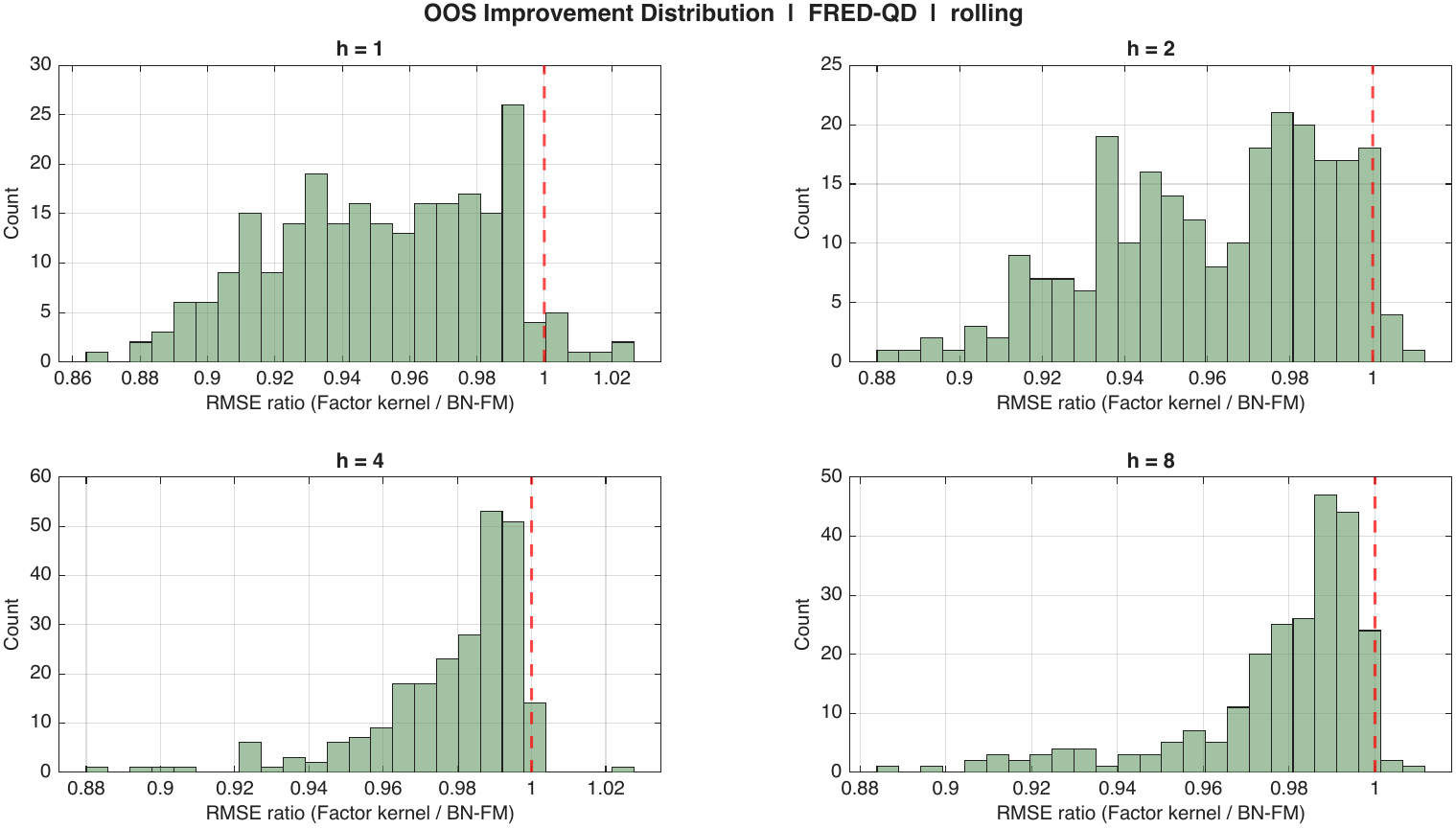}
    \caption{Distribution of the relative improvement of the factor kernel over the
             FM benchmark across all FRED-QD variables, by forecast horizon
             ($h = 1, 2, 4, 8$ quarters; one panel per horizon).
             The horizontal axis shows $\mathrm{RMSE}_{\mathrm{FK}} /
             \mathrm{RMSE}_{\mathrm{FM}}$; values to the left of the dashed vertical
             line (red) indicate that the factor kernel achieves a lower MSFE than the
             FM benchmark.
             MSFE is computed over the out-of-sample evaluation period using a rolling
             estimation window of $W = 40$ quarters,
             Bai--Ng factor selection, and rolling standardization. The quarterly data
             place the estimation in the interpolation regime ($N \gg W$) throughout.}
    \label{fig:oos_hist_QD}
\end{figure}

\begin{figure}[h!]
    \centering
    \includegraphics[width=\linewidth]{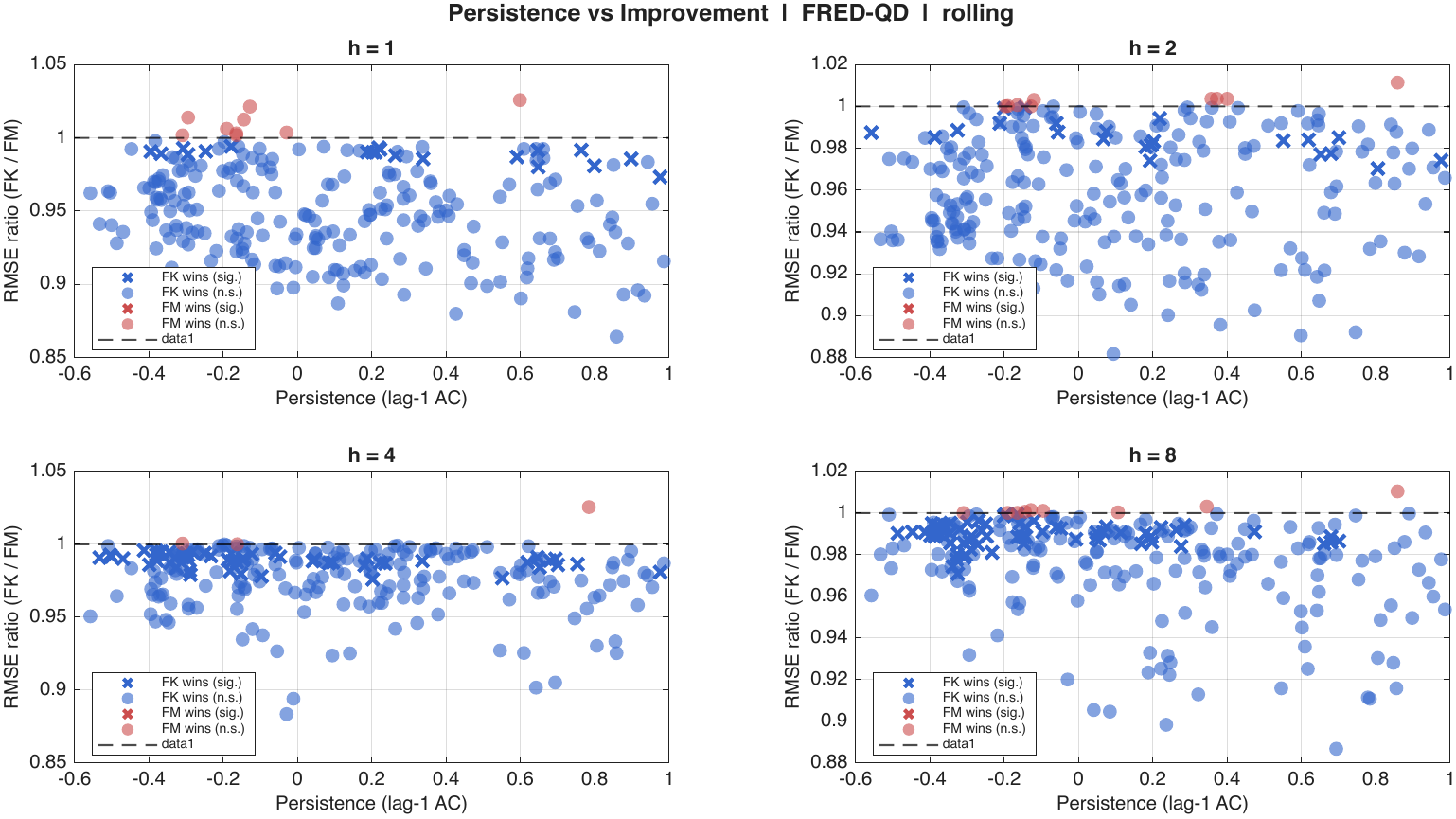}
    \caption{Target persistence vs.\ relative improvement of the factor kernel
             over the FM benchmark, by forecast horizon ($h = 1, 2, 4, 8$ quarters).
             Each point represents one FRED-QD target variable.
             The horizontal axis shows the lag-1 autocorrelation of the (standardized)
             target series, computed on the second half of the available sample.
             The vertical axis shows $\mathrm{RMSE}_{\mathrm{FK}} /
             \mathrm{RMSE}_{\mathrm{FM}}$.
             Blue markers: factor kernel wins; red markers: FM wins.
             Crosses indicate statistical significance at the 5\% level by the Diebold--Mariano test.
             The dashed horizontal line at 1 marks the break-even point.
             Variables are standardized using rolling training-window moments;
             MSFE is accumulated over the out-of-sample evaluation period using a rolling
             estimation window of $W = 40$ quarters.}
    \label{fig:oos_ac1_QD}
\end{figure}

\medskip 

\noindent Figure~\ref{fig:oos_ac1_QD} shows the relationship between persistence and relative out-of-sample performance. Target persistence emerges as the primary predictor of FK advantage, with the correlation between persistence and RMSE ratio strengthening from $-0.17$ at $h=1$ to $-0.32$ at $h=8$. Highly persistent targets benefit more from the factor kernel's state-dependent pooling, consistent with the pattern observed in the monthly data. The relationship is strongest at longer horizons, where the kernel's ability to identify economically similar states becomes most valuable.

\subsection{Subsample analysis}\label{sec:subsample}

To assess whether the factor kernel's advantage is stable across different economic regimes, we partition the 56-year evaluation period into three 20-year windows: 1969--1989 (pre-Great Moderation), 1989--2009 (Great Moderation and Financial Crisis), and 2009--2025 (post-Crisis era). For each target, horizon, and time period, we recompute the MSFE using only forecasts whose target dates fall within that period, and we test for statistical significance using the Diebold--Mariano test applied to the subsample of forecast errors.

Figure~\ref{fig:subsample_MD} shows the results for FRED-MD. Each panel corresponds to one forecast horizon, and each point represents one target variable. The horizontal axis reports target persistence (lag-1 autocorrelation), and the vertical axis reports the RMSE ratio $\mathrm{RMSE}_{\mathrm{FK}} / \mathrm{RMSE}_{\mathrm{FM}}$ computed over the specified time period. Points are color-coded by period: orange for 1969--1989, blue for 1989--2009, and green for 2009--2025. Filled circles indicate that the difference is not statistically significant at the 5\% level; crosses indicate statistical significance.

\begin{figure}[h!]
	\centering
	\includegraphics[width=\linewidth]{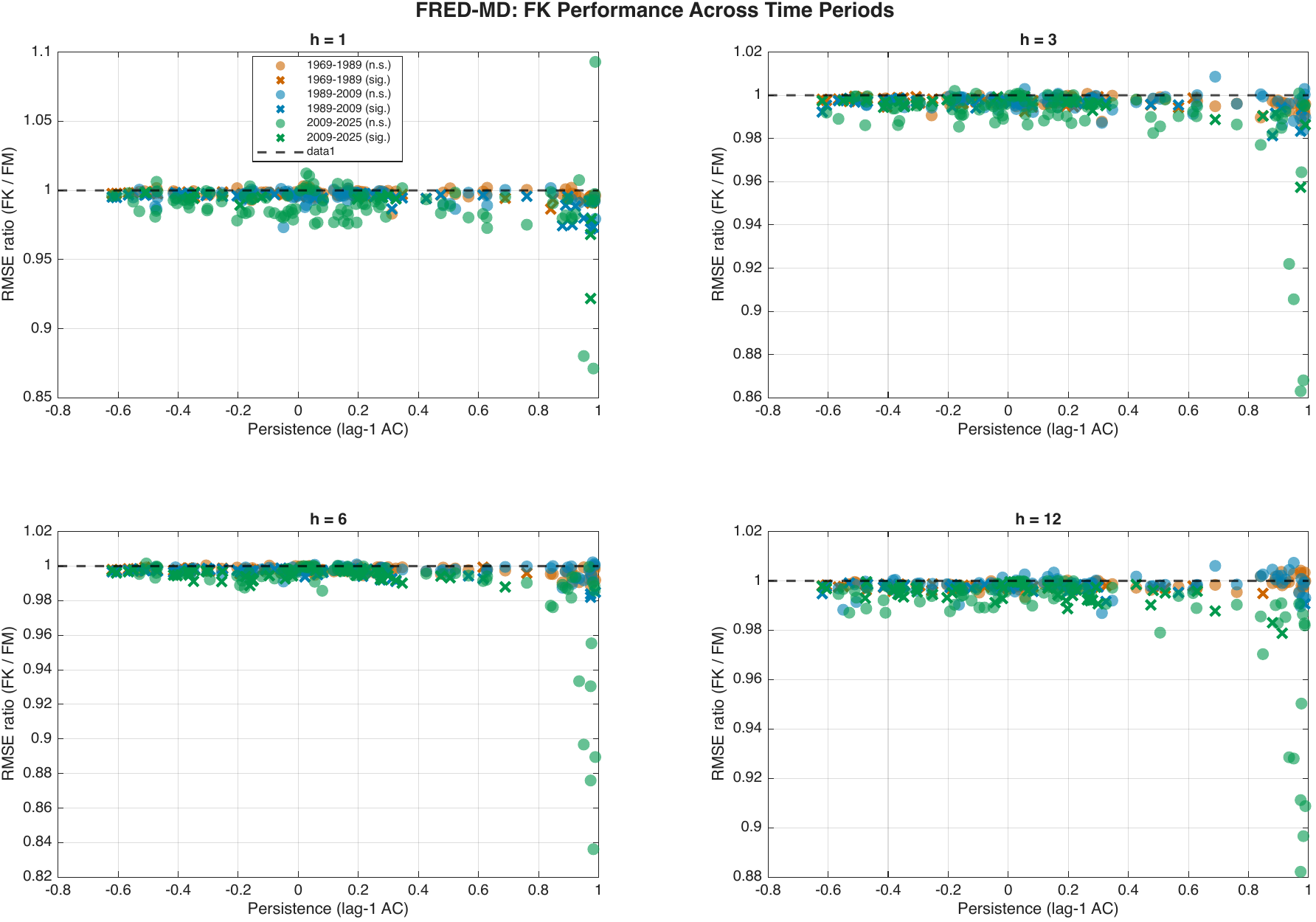}
	\caption{Subsample analysis for FRED-MD: Persistence vs.\ relative improvement
		of the factor kernel over the FM benchmark, by forecast horizon and time period.
		Each point represents one target variable evaluated over a specific 20-year window.
		Orange: 1969--1989 (pre-Great Moderation); Blue: 1989--2009 (Great Moderation
		and Financial Crisis); Green: 2009--2025 (post-Crisis era).
		Filled circles: not statistically significant at 5\% level (DM test);
		Crosses: statistically significant.
		The dashed horizontal line at 1 marks the break-even point.}
	\label{fig:subsample_MD}
\end{figure}

The results reveal that the factor kernel's advantage is broadly stable across time periods, with the negative relationship between persistence and RMSE ratio evident in all three subsamples. At $h=1$, the FK win rates are 91\%, 94\%, and 89\% for the three periods respectively, indicating consistent performance over time. The mean improvements are 0.3\%, 0.5\%, and 1.1\%, suggesting that the magnitude of FK's advantage has grown in the most recent period. Statistical significance is lower in recent periods, with 12\% of wins significant in 2009--2025 compared to 32\% in 1969--1989, reflecting the shorter sample period and lower volatility in the post-Crisis era.

Figure~\ref{fig:subsample_QD} presents the corresponding results for FRED-QD. The quarterly data, which place the estimation in the interpolation regime throughout ($N \gg W$), show similar patterns. The FK win rates are 87\%, 85\%, and 95\% for the three periods at $h=1$, with mean improvements of 0.7\%, 1.2\%, and 6.7\%. The persistence-based predictability of FK advantage remains strong across all periods, with correlations ranging from $-0.17$ to $-0.32$. Notably, the mean improvement is substantially larger in the most recent period (2009--2025), suggesting that the factor kernel's state-dependent pooling has become increasingly valuable in the post-Crisis era.

\begin{figure}[h!]
	\centering
	\includegraphics[width=\linewidth]{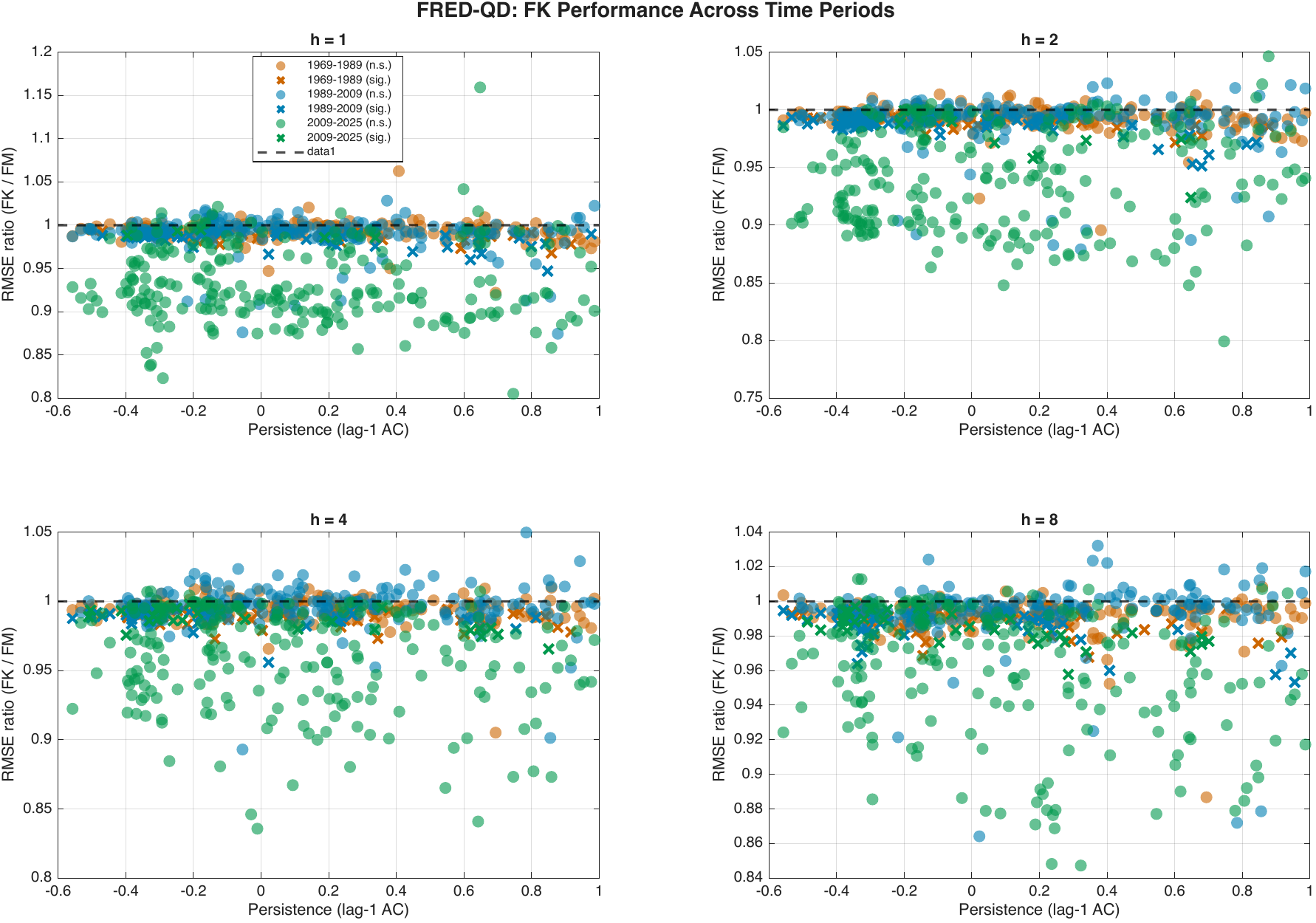}
	\caption{Subsample analysis for FRED-QD: Persistence vs.\ relative improvement
		of the factor kernel over the FM benchmark, by forecast horizon and time period.
		Each point represents one target variable evaluated over a specific 20-year window.
		Orange: 1969--1989 (pre-Great Moderation); Blue: 1989--2009 (Great Moderation
		and Financial Crisis); Green: 2009--2025 (post-Crisis era).
		Filled circles: not statistically significant at 5\% level (DM test);
		Crosses: statistically significant.
		The dashed horizontal line at 1 marks the break-even point.
		The quarterly data place the estimation in the interpolation regime ($N \gg W$)
		throughout all periods.}
	\label{fig:subsample_QD}
\end{figure}

Overall, the subsample analysis confirms that the factor kernel's systematic outperformance is not driven by a particular economic regime or historical episode. The advantage persists across the high-volatility pre-Great Moderation period, the relatively stable Great Moderation, the Financial Crisis, and the post-Crisis era. This robustness across diverse macroeconomic conditions strengthens the case that the factor kernel's state-dependent pooling mechanism provides a genuine improvement over linear factor models, rather than exploiting regime-specific features of the data.

\section{Conclusion}\label{sec:conclusion}

This paper makes three contributions. First, it documents clear double-descent patterns in macroeconomic forecasting data. In the in-sample motivating exercise, the MSFE rises near interpolation and then falls in the overparameterized region, with meaningful differences across variables depending on how strongly each series aligns with the common-factor structure. This provides direct empirical evidence that the double descent phenomenon is relevant in macro panels and not only in canonical machine-learning datasets.

Second, it develops a theoretical link between benign overfitting and factor models. Under exact factor structure, the spectral conditions for benign overfitting are naturally satisfied. Under approximate factor structure, the same mechanism can still operate, but its strength depends on the concentration of the idiosyncratic spectrum. This clarifies when interpolation can be benign in macroeconomic environments and why the answer is target dependent.

Third, it shows that synthetic data augmentation has a precise kernel interpretation and therefore yields an implementable forecasting method. The key insight is that benign overfitting in factor models works by implicitly defining a factor-structured kernel: the augmentation strategy converges to kernel ridge regression with a closed-form kernel that measures temporal proximity through the estimated factor structure rather than calendar time. This kernel can be implemented directly without generating synthetic data, avoiding computational cost and Monte Carlo variability.

In the out-of-sample rolling-window evaluation, the factor-kernel model systematically outperforms the standard Stock-Watson factor model across both monthly and quarterly datasets at all horizons, with win rates exceeding 90\%. The improvements are modest but highly consistent: mean gains are approximately 1\% in the monthly data (underparameterized regime) and 5\% in the quarterly data (interpolation regime).
Statistical significance increases substantially with forecast horizon, with approximately half of wins becoming statistically significant at longer horizons in the monthly data and 30\% in the quarterly data. The evaluation period spans approximately 56 years (1969-2025), providing robust evidence of the factor kernel's performance across diverse macroeconomic conditions. The quarterly results provide direct evidence that the factor kernel performs well in the interpolation regime, where the minimum-norm estimator fits training data exactly yet generalizes effectively when the predictor covariance has the spectral structure implied by factor models.

Across both datasets, we find that highly persistent target series benefit substantially more from the factor kernel's state-dependent pooling, which identifies economically similar states across time regardless of calendar distance. This suggests that the kernel's ability to pool information across similar economic conditions is particularly valuable when the target exhibits strong temporal dependence.

An interesting direction for future research concerns the relationship between signal alignment and forecast performance. While the BLLT theory predicts that targets with signal concentrated in the leading eigenspace should benefit most from the factor kernel, the empirical relationship in our data is weak. Understanding in which situations the factor kernel outperforms linear methods, and reconciling the role of signal alignment versus state-dependent pooling in driving forecasting performance, remains an open question.

\newpage
\appendix

\renewcommand{\thesection}{\Alph{section}}
\makeatletter
\renewcommand\@seccntformat[1]{%
  \ifnum\pdfstrcmp{#1}{section}=0
    Appendix~\csname the#1\endcsname\quad
  \else
    \csname the#1\endcsname\quad
  \fi}
\makeatother

\section{Proofs}\label{sec:appendix}

\subsection{Derivation of variance terms}\label{proof:bias_variance}

We derive the explicit forms of the variance component in equation~\eqref{eq:risk_bv} for both the underparameterized ($N < T$) and overparameterized ($N > T$) regimes. Recall that the MSFE decomposes as
\[
    \mathrm{MSFE} = \underbrace{\E_{\mathbf{X}}\!\left[\|\boldsymbol{\Sigma}^{1/2}(\E[\hat{\boldsymbol{\beta}} \mid \mathbf{X}] - \boldsymbol{\beta}^*)\|^2\right]}_{\text{Bias}^2} + \underbrace{\E_{\mathbf{X}}\!\left[\tr\!\big(\boldsymbol{\Sigma}\,\Var(\hat{\boldsymbol{\beta}} \mid \mathbf{X})\big)\right]}_{\text{Variance}} + \sigma^2.
\]

\paragraph{Variance term.}
We define the conditional variance contribution as
\[
\text{Variance} = \E\!\left[(\hat{\boldsymbol{\beta}}-\E[\hat{\boldsymbol{\beta}}\mid \mathbf{X}])^{\!\top}\boldsymbol{\Sigma}\,(\hat{\boldsymbol{\beta}}-\E[\hat{\boldsymbol{\beta}}\mid \mathbf{X}])\,\middle|\,\mathbf{X}\right].
\]
We will repeatedly use the identity:
\[
\E\!\left[\boldsymbol{\varepsilon}'A\,\boldsymbol{\varepsilon}\,\middle|\,\mathbf{X}\right]= \sigma^{2}\,\tr(A).
\]

\paragraph{Case $N<T$.}
From $\hat{\boldsymbol{\beta}}-\E[\hat{\boldsymbol{\beta}}\mid \mathbf{X}]=(\mathbf{X}'\mathbf{X})^{-1}\mathbf{X}'\boldsymbol{\varepsilon}$, we have
\begin{align*}
\text{Variance}&= \E\!\left[\boldsymbol{\varepsilon}'\mathbf{X}(\mathbf{X}'\mathbf{X})^{-1}\boldsymbol{\Sigma}(\mathbf{X}'\mathbf{X})^{-1}\mathbf{X}'\boldsymbol{\varepsilon}\,\middle|\,\mathbf{X}\right]\\
&= \sigma^{2}\,\tr\!\big(\mathbf{X}(\mathbf{X}'\mathbf{X})^{-1}\boldsymbol{\Sigma}(\mathbf{X}'\mathbf{X})^{-1}\mathbf{X}'\big)\\
&= \sigma^{2}\,\tr\!\big(\boldsymbol{\Sigma}(\mathbf{X}'\mathbf{X})^{-1}\big),
\end{align*}
where the last equality follows from the cyclicity of trace.

\paragraph{Case $N>T$.}
From $\hat{\boldsymbol{\beta}}-\E[\hat{\boldsymbol{\beta}}\mid \mathbf{X}]=\mathbf{X}'(\mathbf{X}\mathbf{X}')^{-1}\boldsymbol{\varepsilon}$, we have
\begin{align*}
\text{Variance}&= \E\!\left[\boldsymbol{\varepsilon}'(\mathbf{X}\mathbf{X}')^{-1}\mathbf{X}\boldsymbol{\Sigma} \mathbf{X}'(\mathbf{X}\mathbf{X}')^{-1}\boldsymbol{\varepsilon}\,\middle|\,\mathbf{X}\right]\\
&= \sigma^{2}\,\tr\!\big((\mathbf{X}\mathbf{X}')^{-1}\mathbf{X}\boldsymbol{\Sigma} \mathbf{X}'(\mathbf{X}\mathbf{X}')^{-1}\big)\\
&= \sigma^{2}\,\tr\!\big(\boldsymbol{\Sigma}\,\mathbf{X}'(\mathbf{X}\mathbf{X}')^{-2}\mathbf{X}\big),
\end{align*}
where the last equality again follows from cyclicity of trace.

Averaging over $\mathbf{X}$ yields:
\[
\E_{\mathbf{X}}[\text{Variance}] =
\begin{cases}
\sigma^{2}\,\E_{\mathbf{X}}\!\big[\tr(\boldsymbol{\Sigma}(\mathbf{X}'\mathbf{X})^{-1})\big], & N<T,\\[6pt]
\sigma^{2}\,\E_{\mathbf{X}}\!\big[\tr(\boldsymbol{\Sigma}\,\mathbf{X}'(\mathbf{X}\mathbf{X}')^{-2}\mathbf{X})\big], & N>T.
\end{cases}
\]

These are the variance expressions used in equations~\eqref{eq:bias_OP} and~\eqref{eq:var_OP} in Section~\ref{sec:double_descent}.

\subsection{Proof of Lemma~\ref{lemma:spiked_spectrum}}\label{proof:lemma_spiked_spectrum}

\begin{proof}
Since $\bm{\Lambda}$ is $N \times k$, the matrix 
$\bm{\Lambda}\bm{\Lambda}'$ has rank at most $k$, so 
$\lambda_j(\bm{\Lambda}\bm{\Lambda}') = 0$ for all $j > k$. Its 
$k$ nonzero eigenvalues coincide with those of the 
$k \times k$ matrix $\bm{\Lambda}'\bm{\Lambda}$, whose $(j,l)$ 
entry is $\sum_{i=1}^N \lambda_{ij}\lambda_{il}$. 
Assumption~\ref{ass:factor}(iv) ensures 
$\bm{\Lambda}'\bm{\Lambda}/N \to D$ with $D$ positive definite, so 
all $k$ eigenvalues of $\bm{\Lambda}'\bm{\Lambda}$ are $O(N)$, 
including the smallest. Furthermore, Assumption~\ref{ass:factor}(ii) and Assumption~\ref{ass:factor_approx}(ii$'$) imply that $\bar{\psi}^2 \;=\; O(1)$. Weyl's inequality applied to 
$\bm{\Sigma} = \bm{\Lambda}\bm{\Lambda}' + \bm{\Psi}$ then gives the spectral 
gap: for $j \leq k$,
$\lambda_j(\bm{\Sigma}) \;\geq\; 
    \lambda_j(\bm{\Lambda}\bm{\Lambda}') + \lambda_N(\bm{\Psi}) 
    \;=\; O(N),
$
while for $j > k$,
$\underline{\psi}^2 \;=\; \lambda_N(\bm{\Psi}) 
    \;\leq\; \lambda_j(\bm{\Sigma}) \;\leq\; 
    \lambda_{j}(\bm{\Lambda}\bm{\Lambda}') + \lambda_1(\bm{\Psi}) 
    \;=\; 0 + \bar{\psi}^2 \;=\; O(1)$.
\end{proof}

\subsection{Proof of Theorem~\ref{prop:factor}}\label{proof:prop_factor}

\begin{proof}
    \emph{(i) Overall effective rank.} 
    
    By Weyl's inequality~\eqref{eq:weyl}, $\lambda_1(\bm{\Sigma}) \geq \lambda_1(\bm{\Lambda}\bm{\Lambda}')$. Since $\lambda_1(\bm{\Lambda}\bm{\Lambda}') = O(N)$ by Assumption~\ref{ass:factor}(iv) or Assumption~\ref{ass:factor_approx}(iv), we have $\lambda_1(\bm{\Sigma}) = O(N)$. For the trace, $\tr(\bm{\Sigma}) = \tr(\bm{\Lambda}\bm{\Lambda}') + \tr(\bm{\Psi}) = O(N) + O(N) = O(N)$ under both assumptions. Therefore $r_0 = \tr(\bm{\Sigma})/\lambda_1 = O(N)/O(N) = O(1)$. Since $T \to \infty$ by assumption, $r_0/T \to 0$.

    \emph{(ii) Low-dimensional head.} 
    
    Define $\psi_{\min}^2$ and $\psi_{\max}^2$ as follows: under Assumption~\ref{ass:factor} (exact model), $\psi_{\min}^2 = \min_j \psi_j^2$ and $\psi_{\max}^2 = \max_j \psi_j^2$ are the minimum and maximum diagonal entries of $\bm{\Psi}$; under Assumption~\ref{ass:factor_approx} (approximate model), $\psi_{\min}^2 = \lambda_{\min}(\bm{\Psi})$ and $\psi_{\max}^2 = \lambda_{\max}(\bm{\Psi})$ are the minimum and maximum eigenvalues of $\bm{\Psi}$. In both cases, $\underline{\psi}^2 \leq \psi_{\min}^2 \leq \psi_{\max}^2 \leq \bar{\psi}^2$ by assumption.
    
    By Weyl's inequality~\eqref{eq:weyl}, for $j \leq k$:
    \[
        \lambda_j(\bm{\Sigma}) \;\geq\; \psi_{\min}^2 + \lambda_j(\bm{\Lambda}\bm{\Lambda}') \;\geq\; \underline{\psi}^2 + \lambda_j(\bm{\Lambda}\bm{\Lambda}').
    \]
    Since $\lambda_k(\bm{\Lambda}\bm{\Lambda}') = O(N)$ (the $k$-th eigenvalue of a rank-$k$ matrix with $\bm{\Lambda}'\bm{\Lambda}/N \to \bm{\Sigma}_\Lambda$ positive definite), we obtain $\lambda_j(\bm{\Sigma}) = O(N)$ for all $j \leq k$. For $j > k$, since $\lambda_j(\bm{\Lambda}\bm{\Lambda}') = 0$, Weyl's inequality gives
    \[
        \lambda_j(\bm{\Sigma}) \;\leq\; \psi_{\max}^2 + 0 \;\leq\; \bar{\psi}^2 = O(1).
    \]
    
    This establishes the spectral gap: $\lambda_k(\bm{\Sigma}) = O(N)$ while $\lambda_{k+1}(\bm{\Sigma}) = O(1)$. For $j < k$, both $\sum_{i>j}\lambda_i$ and $\lambda_{j+1}$ are $O(N)$, so $r_j = O(1) < bT$ for large $T$. At $j = k$, the numerator $\sum_{i>k}\lambda_i = (N-k) \cdot O(1) = O(N)$ while the denominator $\lambda_{k+1} = O(1)$, giving $r_k = O(N) \geq bT$. Therefore $k^* = k$ and $k^*/T \to 0$.

    \emph{(iii) Tail flatness.} 
    
    The argument for condition (iii) is identical under both assumptions, relying on the interlacing inequalities for rank-$k$ perturbations.
    
    Case~(a): Under Assumption~\ref{ass:factor} with $\bm{\Psi} = \psi^2 I_N$, Weyl's inequality gives $\lambda_j(\bm{\Sigma}) = \lambda_j(\bm{\Lambda}\bm{\Lambda}') + \psi^2$ for all $j$. Since $\lambda_j(\bm{\Lambda}\bm{\Lambda}') = 0$ for $j > k$, all tail eigenvalues equal $\psi^2$ and $\mathcal{C} = 1$. Under Assumption~\ref{ass:factor_approx}, if all eigenvalues of $\bm{\Psi}$ satisfy $\mu_j = \psi^2$, then by the interlacing inequalities (stated below), $\lambda_j(\bm{\Sigma}) = \psi^2$ for all $j > k$ and again $\mathcal{C} = 1$.

    Cases~(b)--(c): When $\bm{\Psi}$ is not scalar (exact model) or its eigenvalues are not all equal (approximate model), the tail eigenvalues of $\bm{\Sigma}$ no longer coincide with the eigenvalues of $\bm{\Psi}$. However, the interlacing inequalities\footnote{If $B$ is a positive semidefinite matrix of rank at most $k$, the eigenvalues of $A + B$ interlace with those of $A$ with slack $k$: $\lambda_{j+k}(A+B) \leq \lambda_j(A) \leq \lambda_{j-k}(A+B)$. See \citet{horn2012matrix}, Theorem~4.3.6.} for the rank-$k$ perturbation $\bm{\Lambda}\bm{\Lambda}'$ give
    \[
        \psi_{(j+k)}^2 \;\leq\; \lambda_{j+k}(\bm{\Sigma}) \;\leq\; \psi_{(j)}^2, \qquad j = 1, \ldots, N-k,
    \]
    where $\psi_{(1)}^2 \geq \cdots \geq \psi_{(N)}^2$ are the ordered eigenvalues of $\bm{\Psi}$ (which coincide with the ordered diagonal entries $\psi_1^2, \ldots, \psi_N^2$ under the exact factor model). Since all eigenvalues of $\bm{\Psi}$ lie in $[\underline{\psi}^2, \bar{\psi}^2]$ under both assumptions, the interlacing bounds imply that every tail eigenvalue of $\bm{\Sigma}$ also lies in this interval:
    \[
        \underline{\psi}^2 \;\leq\; \lambda_j(\bm{\Sigma}) \;\leq\; \bar{\psi}^2, \qquad j = k+1, \ldots, N.
    \]
    For any set of values confined to $[\underline{\psi}^2, \bar{\psi}^2]$, the concentration ratio satisfies
    \[
        1 \;\leq\; \mathcal{C} \;=\; \frac{\overline{\lambda_{\mathrm{tail}}^2}}{(\bar{\lambda}_{\mathrm{tail}})^2} \;\leq\; \frac{\bar{\psi}^2}{\underline{\psi}^2},
    \]
    where the lower bound is Jensen's inequality and the upper bound follows from $\overline{\lambda_{\mathrm{tail}}^2} \leq \bar{\psi}^4$ and $\bar{\lambda}_{\mathrm{tail}} \geq \underline{\psi}^2$. Thus $\mathcal{C}$ is controlled by the same primitives $\underline{\psi}^2$ and $\bar{\psi}^2$ under both assumptions, and the flatness condition $T\mathcal{C}/(N-k) \to 0$ holds under the conditions on $\mathcal{C}$ stated in the theorem.
\end{proof}

\subsection{Proof of Lemma~\ref{lem:kernel_duality}}

\begin{proof}\label{proof:kernel_duality}
    The minimum-norm interpolator is 
    $\hat{\beta} = \bZ'(\bZ\bZ')^{-1}y$, which 
    can be verified to satisfy 
    $\bZ\hat{\beta} = y$ and to have minimal 
    $\ell_2$ norm among all solutions (see, e.g., 
    the pseudoinverse characterization in 
    Section~\ref{sec:setup}). The prediction for 
    a new point is
    \begin{align*}
        \hat{y}_{\textup{new}} 
        &= z_{\textup{new}}'\hat{\beta} \\
        &= z_{\textup{new}}'\bZ'(\bZ\bZ')^{-1}y \\
        &= k'(\bZ\bZ')^{-1}y,
    \end{align*}
    where $k = \bZ z_{\textup{new}} \in 
    \mathbb{R}^T$ has entries 
    $k_s = \bZ_s' z_{\textup{new}} = 
    z_{\textup{new}}'\bZ_s'$, i.e., the inner 
    product between the new observation and the 
    $s$-th training observation in the augmented 
    feature space. This is the standard kernel 
    interpolation formula: the Gram matrix 
    $\bZ\bZ'$ plays the role of the kernel matrix 
    evaluated at training points, and $k$ is the 
    kernel evaluated between the new point and 
    each training point.
\end{proof}

\subsection{Proof of 
Proposition~\ref{thm:synth_kernel}}\label{proof:synth_kernel}

We first establish the expected Gram matrix, then 
apply concentration and the Woodbury identity to 
obtain the kernel dual.

\begin{lemma}[Expected Gram matrix of 
synthetic-augmented system]
\label{lem:synth_gram}
    The expected Gram matrix 
    $\E_*[\bZ\bZ']$ satisfies
    \[
        \E_*[(\bZ\bZ')_{ts}] \;=\; 
        (XX')_{ts} + B\Big(
        \hat{\mathbf{f}}_t'\hat{\bm{\Lambda}}'\hat{\bm{\Lambda}}
        \hat{\mathbf{f}}_s + \textup{tr}(\hat{\bm{\Psi}}) 
        \cdot \mathbf{1}[t = s]\Big).
    \]
\end{lemma}

\begin{proof}
    The expectation $\mathbb{E}_*$ is conditional 
    on the observed data and the estimated model 
    parameters $\hat{\bm{\Lambda}}$, $\hat{\bm{\Psi}}$, 
    $\hat{\mathbf{f}}_1, \ldots, \hat{\mathbf{f}}_T$. Since 
    $\mathbf{x}_t^{*(b)} = \hat{\bm{\Lambda}}\hat{\mathbf{f}}_t + 
    \mathbf{e}_t^{*(b)}$ with $\mathbf{e}_t^{*(b)}$ independent 
    across $t$ and $b$,
    \begin{align*}
        \E_*[\mathbf{x}_t^{*(b)\prime}\mathbf{x}_s^{*(b)}] 
        &= \E_*[(\hat{\bm{\Lambda}}\hat{\mathbf{f}}_t + 
        \mathbf{e}_t^{*(b)})'(\hat{\bm{\Lambda}}\hat{\mathbf{f}}_s + 
        \mathbf{e}_s^{*(b)})] \\
        &= \hat{\mathbf{f}}_t'\hat{\bm{\Lambda}}'\hat{\bm{\Lambda}}
        \hat{\mathbf{f}}_s + \E_*[\mathbf{e}_t^{*(b)\prime}
        \mathbf{e}_s^{*(b)}] \\
        &= \hat{\mathbf{f}}_t'\hat{\bm{\Lambda}}'\hat{\bm{\Lambda}}
        \hat{\mathbf{f}}_s + \textup{tr}(\hat{\bm{\Psi}}) 
        \cdot \mathbf{1}[t = s],
    \end{align*}
    where the cross-terms 
    $\hat{\mathbf{f}}_t'\hat{\bm{\Lambda}}'\E_*[\mathbf{e}_s^{*(b)}] = 0$ 
    vanish by $\E_*[\mathbf{e}_s^{*(b)}] = 0$, and 
    $\E_*[\mathbf{e}_t^{*(b)\prime}\mathbf{e}_s^{*(b)}] = 
    \textup{tr}(\hat{\bm{\Psi}}) \cdot 
    \mathbf{1}[t=s]$ by independence across $t$. 
    Summing the fixed term $\bm{X}\bm{X}'$ and $B$ independent
    replications gives the result.
\end{proof}

\begin{proof}[Proof of 
Proposition~\ref{thm:synth_kernel}]

    \emph{Step 1: Gram matrix concentration.}
    Since $\bZ = [\bm{X} \mid \bm{X}^{*(1)} \mid \cdots 
    \mid \bm{X}^{*(B)}]$, the Gram matrix decomposes as
    \[
        \bZ\bZ' \;=\; \bm{X}\bm{X}' + \sum_{b=1}^B
        \bm{X}^{*(b)}\bm{X}^{*(b)\prime}.
    \]
    The $B$ replications are independent 
    conditional on the data. By the law of large 
    numbers and Lemma~\ref{lem:synth_gram},
    \[
        \frac{1}{B}\sum_{b=1}^B 
        \bm{X}^{*(b)}\bm{X}^{*(b)\prime} \;\to\; 
        \hat{\mathbf{F}}\hat{\bm{\Lambda}}'\hat{\bm{\Lambda}}
        \hat{\mathbf{F}}' + \textup{tr}(\hat{\bm{\Psi}})\,I_T,
    \]
    where $\hat{\mathbf{F}} \in \mathbb{R}^{T \times k}$ 
    is the matrix with rows 
    $\hat{\mathbf{f}}_1', \ldots, \hat{\mathbf{f}}_T'$. Since 
    $\bm{X}\bm{X}'/B \to 0$, the Gram matrix satisfies
    \[
        \bZ\bZ' \;=\; B\Big(
        \hat{\mathbf{F}}\hat{\bm{\Lambda}}'\hat{\bm{\Lambda}}
        \hat{\mathbf{F}}' + \textup{tr}(\hat{\bm{\Psi}})
        \,I_T\Big) + o_p(B).
    \]

    \emph{Step 2: Kernel duality.}
    By Lemma~\ref{lem:kernel_duality}, the 
    minimum-norm interpolator predicts via 
    $\hat{y}_{\textup{new}} = 
    k'(\bZ\bZ')^{-1}y$, which depends on $\bZ$ 
    only through the Gram matrix and the kernel vector $\mathbf{k}$. Substituting 
    the concentrated Gram matrix:
    \[
        \bZ\bZ' \;\approx\; 
        B\,K_{\textup{synth}} + 
        B\,\textup{tr}(\hat{\bm{\Psi}})\,I_T,
    \]
    where $K_{\textup{synth}}(t,s) = 
    \hat{\mathbf{f}}_t'\hat{\bm{\Lambda}}'\hat{\bm{\Lambda}}
    \hat{\mathbf{f}}_s$. The kernel evaluation vector 
    $k$ satisfies $k \approx B\, 
    k_{\textup{synth}}$ by the same 
    concentration argument. Therefore
    \begin{align*}
        \hat{y}_{\textup{new}} 
        &\approx (B\,k_{\textup{synth}})'
        (B\,K_{\textup{synth}} + 
        B\,\textup{tr}(\hat{\bm{\Psi}})\,I_T)^{-1}y \\
        &= k_{\textup{synth}}'
        (K_{\textup{synth}} + 
        \textup{tr}(\hat{\bm{\Psi}})\,I_T)^{-1}y,
    \end{align*}
    where the factor $B$ cancels. This is kernel 
    ridge regression with kernel 
    $K_{\textup{synth}}$ and ridge parameter 
    $\lambda = \textup{tr}(\hat{\bm{\Psi}})$. The 
    matrix $K_{\textup{synth}} = 
    \hat{\mathbf{F}}\hat{\bm{\Lambda}}'\hat{\bm{\Lambda}}\hat{\mathbf{F}}'$ 
    is positive semidefinite of rank at most $k$, 
    so $K_{\textup{synth}} + \lambda I_T$ is 
    invertible for any $\lambda > 0$.
    
\end{proof}

\section{Double Descent with FRED-QD}\label{app:fred_qd}

This appendix presents the double-descent illustration using the FRED-QD quarterly dataset, which complements the monthly analysis in Section~\ref{sec:motivation}. The quarterly data provide a different perspective on the interpolation regime: with $N \gg T$ throughout, the estimation is always in the overparameterized region.

\begin{figure}[t!]
    \centering
    \includegraphics[width=1\linewidth]{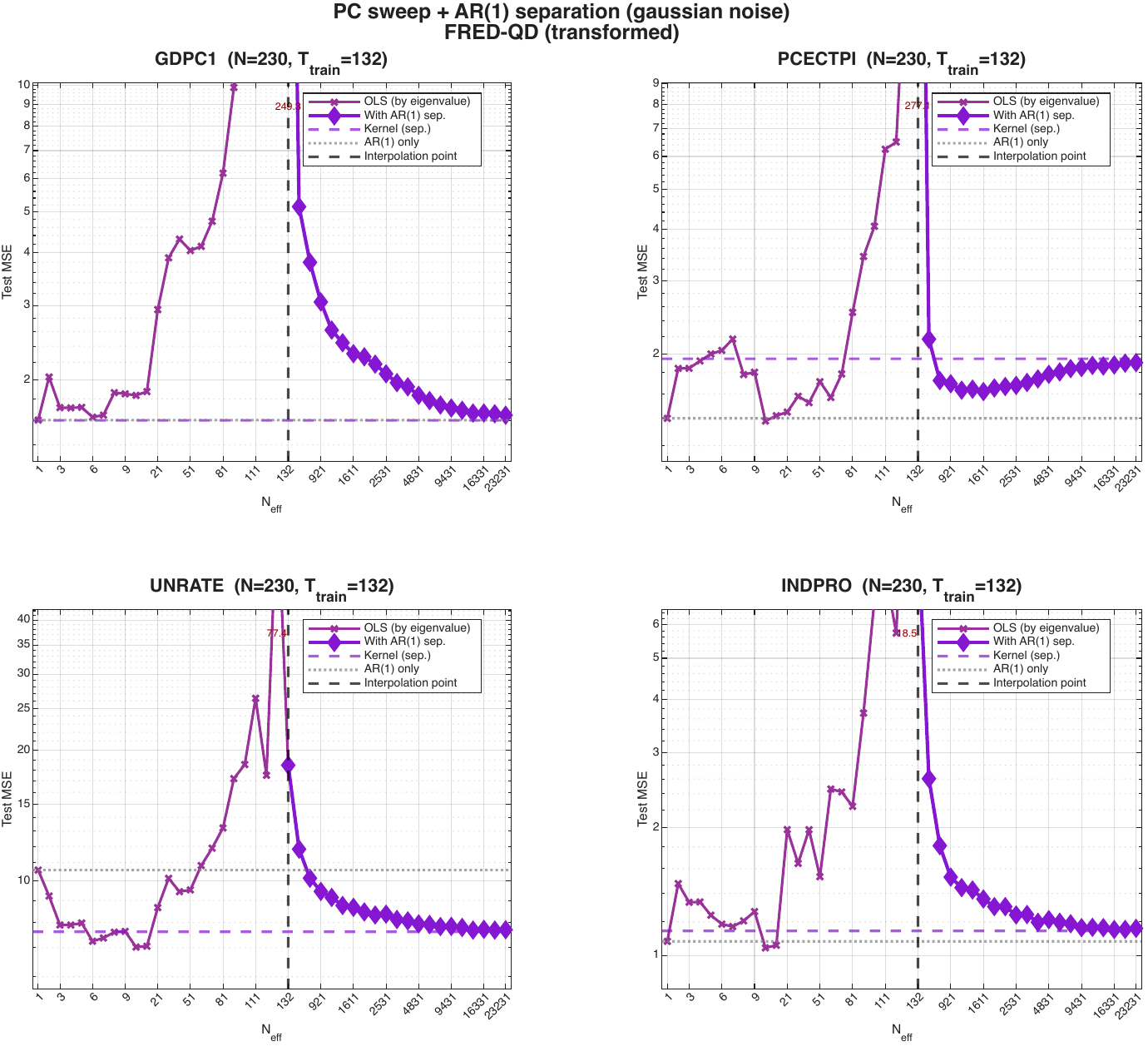}
    \caption{Double-descent MSFE curves for four FRED-QD targets.
    The thin purple line traces the PC sweep from $r=1$ to $r=N$;
    the thick purple line traces results based on $B$ synthetic copies of the data $\bm X$.
    The dashed vertical line is the interpolation threshold ($N_{\mathrm{eff}} = T_{\mathrm{train}}$);
    the purple dashed horizontal line is the asymptotic limit of the estimator, computed from its RKHS (kernel) representation described in Section~\ref{sec:augmentation_andRKHS};
    the gray dotted horizontal line shows the AR(1) baseline.
    The quarterly data place the estimation in the interpolation regime ($N \gg T$) throughout the augmentation path.}
    \label{fig:dd_four_panel_qd}
\end{figure}

We take the FRED-QD quarterly dataset of \citet{mccracken2016fred}, which contains $N_{\mathrm{raw}} = 248$ macroeconomic and financial series spanning 1959:Q1 to 2025:Q1 ($T = 264$ observations); after excluding series with more than 20\% missing values, a typical target uses $N \approx 240$ predictors. We split the whole sample in half, with $T_{\mathrm{train}} = 132$ and $T_{\mathrm{test}} = 132$. All models considered include an AR(1) lag of the target variable as a baseline predictor; we then look at what happens as we add cross-sectional information on top of it.

We begin by taking an AR(1) model (one parameter), and gradually increase model size toward $N \approx 240$ by adding principal components (PCs) of $\bm X$ one at a time, ranked by eigenvalue. This yields an ordered path from AR(1)+1 PC to AR(1)+all $N$ PCs (which is algebraically equivalent to using all $N$ original predictors). The first set of PCs captures the dominant common factors, while later PCs increasingly capture idiosyncratic dynamics. Figure~\ref{fig:dd_four_panel_qd} shows that the resulting test MSE profile traces the usual bias-variance tradeoff.

At $r = N$ the model uses all available variables and we have reached the boundary of the original data. To push into the overparameterized regime, we augment the design matrix with synthetic copies of $\bm X$. Each copy preserves the estimated common component (the factors) but replaces the idiosyncratic errors with fresh Gaussian draws. Appending $B$ such copies creates a design matrix with $N(B+1)$ columns, so the effective number of parameters grows in discrete jumps of $N$. The construction and its theoretical properties are developed in Section~\ref{sec:augmentation_andRKHS}.

In this exercise, we let $N_{\mathrm{eff}}$ denote the number of slope coefficients being estimated (cross-sectional predictors plus the AR(1) lag). Along the PC path, the number of coefficients is $N_{\mathrm{eff}}=r+1$. With augmentation, we end up estimating $N_{\mathrm{eff}}=(B+1)N+1$ parameters. The interpolation threshold is reached when $N_{\mathrm{eff}}=T_{\mathrm{train}}$. With $T_{\mathrm{train}} = 132$ quarters and $N \approx 240$ predictors, we have $N \gg T$ throughout, so even the PC sweep with $r=1$ is already overparameterized, and the augmentation path remains deeply in the interpolation regime.

We apply this procedure to four FRED-QD targets (real GDP, PCE price index, unemployment rate, and industrial production) and plot test MSE against the effective number of parameters. Figure~\ref{fig:dd_four_panel_qd} shows a clear double-descent shape in all four cases: MSE rises sharply near the interpolation threshold and then declines in the overparameterized region.

\section{Illustration: Bias and Variance Under Alternative Spectral Shapes}
\label{sec:bias-spectra}
This appendix illustrates the bias and variance requirements for benign overfitting discussed in Section~\ref{sec:bllt_conditions} through three concrete spectral configurations. Recall that the MSFE decomposes into bias and variance terms. For benign overfitting to occur, both terms must vanish as $N/T \to \infty$. We analyze the minimum-norm estimator in the overparameterized regime $N > T$ under three spectral configurations: (1) isotropic spectrum, (2) geometric decay spectrum, and (3) flat-tail spectrum (factor structure). Throughout, we work in the eigenbasis of $\boldsymbol{\Sigma}$ and assume the signal $\boldsymbol{\beta}^*$ is concentrated in the first $m < T$ coordinates.
We first examine the bias term for all three cases and show that only the flat-tail case achieves vanishing bias. We then verify that the variance term also vanishes for the flat-tail case, confirming that it satisfies both requirements for benign overfitting.

\subsection{Bias term analysis}
The bias term in the MSFE measures how far the expected prediction is from the true signal. In the overparameterized regime $N > T$, the minimum-norm estimator satisfies
\[
\E[\hat{\boldsymbol{\beta}} \mid \mathbf{X}] = \bm{P}_r\boldsymbol{\beta}^*, \qquad
\bm{P}_r := \mathbf{X}^\top(\mathbf{X}\mathbf{X}^\top)^{-1}\mathbf{X},
\]
where $\bm{P}_r$ is the orthogonal projector onto the row space $\mathcal{R} = \text{row}(\mathbf{X}) \subset \mathbb{R}^N$. The estimator recovers only the component of $\boldsymbol{\beta}^*$ that lies in the row space; the orthogonal component $(I - \bm{P}_r)\boldsymbol{\beta}^*$ remains unrecovered. The $\boldsymbol{\Sigma}$-weighted squared bias is
\begin{equation}
\label{eq:bias-general}
\text{Bias}^2(\mathbf{X})
= \big((I - \bm{P}_r)\boldsymbol{\beta}^*\big)^\top\boldsymbol{\Sigma}\big((I - \bm{P}_r)\boldsymbol{\beta}^*\big)
= \sum_{j=1}^N \lambda_j\big(e_j^\top(I - \bm{P}_r)\boldsymbol{\beta}^*\big)^2.
\end{equation}
This weights each unrecovered direction by its variance $\lambda_j$. Since $\mathbf{X}$ is random (drawn from a Gaussian design with covariance $\boldsymbol{\Sigma}$), the row space $\mathcal{R}$ is random, and we study the expected bias $\E[\text{Bias}^2(\mathbf{X})]$ over the distribution of $\mathbf{X}$. The question is: under what spectral configurations does this expectation vanish as $N \to \infty$?

\paragraph{Case 1: Isotropic spectrum (fails ``light enough'' requirement)}

Let $\lambda_j \equiv 1$ for all $j$, so that $\boldsymbol{\Sigma} = I_N$. In this case, the row space of $\mathbf{X}$ is uniformly distributed over all $T$-dimensional subspaces of $\mathbb{R}^N$, and by rotational invariance
\[
\E[\bm{P}_r] = \frac{T}{N}I_N.
\]
Hence the expected squared bias admits the closed form
\begin{equation}
\label{eq:bias-isotropic}
\E[\text{Bias}^2(\mathbf{X})]
= \E\big[(\boldsymbol{\beta}^*)^\top(I - \bm{P}_r)\boldsymbol{\beta}^*\big]
= \Big(1 - \frac{T}{N}\Big)\|\boldsymbol{\beta}^*\|^2, \qquad N > T.
\end{equation}
As $N \to \infty$, the bias converges to $\|\boldsymbol{\beta}^*\|^2 > 0$. 

The tail eigenvalues are not light enough ($\lambda_j = 1$ for all $j > m$). Even though the row space expands and $(I - \bm{P}_r)\boldsymbol{\beta}^*$ shrinks, the $\boldsymbol{\Sigma}$-weighting does not reduce the bias because the unrecovered component lies in directions with unit variance. The tail is not ``light enough''---eigenvalues must be small for the weighted bias to vanish.

\paragraph{Case 2: Geometric decay spectrum (fails ``disperse enough'' requirement)}

Consider a geometrically decaying spectrum
\[
\lambda_j = q^{j-1}, \qquad 0 < q < 1.
\]
No closed-form expression for $\E[\text{Bias}^2(\mathbf{X})]$ is available because $\bm{P}_r$ and $\boldsymbol{\Sigma}$ do not commute.

To characterize the asymptotic behavior, fix an integer $C \geq 1$ and decompose the coordinates into the head $H = \{1, \ldots, m\}$, an early tail block $T_1 = \{m+1, \ldots, m+C\}$, and the remaining tail $T_2 = \{m+C+1, \ldots, N\}$. Since all terms in \eqref{eq:bias-general} are nonnegative,
\[
\E[\text{Bias}^2(\mathbf{X})]
\geq
\E\Bigg[\sum_{j \in T_1}\lambda_j\big(e_j^\top(I - \bm{P}_r)\boldsymbol{\beta}^*\big)^2\Bigg].
\]
For $j \in T_1$ we have $\lambda_j \geq q^{m+C-1}$, which implies
\begin{equation}
\label{eq:bias-geom-lower}
\E[\text{Bias}^2(\mathbf{X})]
\geq
q^{m+C-1}\,\E\Big[\|(I - \bm{P}_r)\boldsymbol{\beta}^*\|_{T_1}^2\Big],
\end{equation}
where $\|v\|_{T_1}^2 := \sum_{j \in T_1}v_j^2$. Under geometric decay, the total variance in the deeper tail $T_2$ is finite:
\[
\sum_{j > m+C}\lambda_j
= \sum_{k=C+1}^\infty q^{m+k-1}
= \frac{q^{m+C}}{1 - q}.
\]
Thus, adding further coordinates in $T_2$ does not introduce a growing number of effective low-variance directions. The residual $(I - \bm{P}_r)\boldsymbol{\beta}^*$ cannot be dispersed over an increasing-dimensional tail, and $\E[\|(I - \bm{P}_r)\boldsymbol{\beta}^*\|_{T_1}^2]$ remains bounded away from zero. Combining this with \eqref{eq:bias-geom-lower} yields
\begin{equation}
\label{eq:bias-geom-order}
\liminf_{N \to \infty}\E[\text{Bias}^2(\mathbf{X})]
\geq
\text{const} \times q^m > 0.
\end{equation}

The tail eigenvalues vanish exponentially fast. Even though they are small (satisfying ``light enough''), there are too few effective directions for the row space to expand and cover the signal support. The tail is not ``disperse enough''---we need many directions with similar small eigenvalues, not exponentially vanishing ones.

\paragraph{Case 3: Flat-tail spectrum / Factor structure (satisfies both requirements)}

Consider a spectrum with a head and a flat tail:
\[
\lambda_j =
\begin{cases}
q^{j-1}, & j \leq m_{\text{head}},\\
\tau, & j > m_{\text{head}},
\end{cases}
\qquad \tau = \tau_{T,N} \downarrow 0,
\]
with $m \leq m_{\text{head}}$ so that the signal is contained in the head.

The squared bias decomposes as
\[
\text{Bias}^2(\mathbf{X})
= \sum_{j \leq m_{\text{head}}}\lambda_j\big((I - \bm{P}_r)\boldsymbol{\beta}^*\big)_j^2
+ \sum_{j > m_{\text{head}}}\tau\big((I - \bm{P}_r)\boldsymbol{\beta}^*\big)_j^2.
\]
Since $(I - \bm{P}_r)$ is an orthogonal projector, $\|(I - \bm{P}_r)\boldsymbol{\beta}^*\|^2 \leq \|\boldsymbol{\beta}^*\|^2$, and therefore
\begin{equation}
\label{eq:bias-flat-upper}
\text{Bias}^2(\mathbf{X})
\leq
\sum_{j \leq m_{\text{head}}}\lambda_j\big((I - \bm{P}_r)\boldsymbol{\beta}^*\big)_j^2
+ \tau\|\boldsymbol{\beta}^*\|^2.
\end{equation}
As $N$ grows, the row space increasingly aligns with the low-dimensional head, so the first term in \eqref{eq:bias-flat-upper} becomes negligible. Moreover, if $\tau = O(T/N)$, then $\tau \to 0$ as $N/T \to \infty$, implying
\[
\E[\text{Bias}^2(\mathbf{X})] \to 0.
\]

The flat tail provides many weak directions (disperse enough) with small eigenvalues (light enough). The row space can expand and progressively align with the head space where $\boldsymbol{\beta}^*$ has support, while the unrecovered residual is dispersed over many low-variance tail directions. The $\boldsymbol{\Sigma}$-weighted bias vanishes because both requirements are satisfied.

This spectral shape arises naturally in factor-structured panels, where the head eigenvalues correspond to common factors and the flat tail corresponds to idiosyncratic components with similar small variances. We now verify that the variance condition also holds for this case.

\subsection{Variance term analysis for flat-tail case}

We now verify that also the variance converges to 0 for the flat-tail (factor structure) case. Recall from Appendix~\ref{proof:bias_variance} that for $N > T$, the variance term is
\[
\text{Variance} = \sigma^2\,\E_{\mathbf{X}}\!\big[\tr(\boldsymbol{\Sigma}\,\mathbf{X}^\top(\mathbf{X}\mathbf{X}^\top)^{-2}\mathbf{X})\big].
\]
For the flat-tail spectrum with $\lambda_j = \tau$ for $j > m_{\text{head}}$, the trace decomposes as
\[
\tr(\boldsymbol{\Sigma}\,\mathbf{X}^\top(\mathbf{X}\mathbf{X}^\top)^{-2}\mathbf{X})
= \sum_{j=1}^{m_{\text{head}}}\lambda_j\,e_j^\top\mathbf{X}^\top(\mathbf{X}\mathbf{X}^\top)^{-2}\mathbf{X} e_j
+ \tau\sum_{j > m_{\text{head}}}e_j^\top\mathbf{X}^\top(\mathbf{X}\mathbf{X}^\top)^{-2}\mathbf{X} e_j.
\]
The key observation is that $\mathbf{X}^\top(\mathbf{X}\mathbf{X}^\top)^{-2}\mathbf{X}$ is positive semidefinite with $\tr(\mathbf{X}^\top(\mathbf{X}\mathbf{X}^\top)^{-2}\mathbf{X}) = \tr((\mathbf{X}\mathbf{X}^\top)^{-1}) = O(1/T)$ as $T \to \infty$. Since the tail contains $N - m_{\text{head}}$ directions, each with eigenvalue $\tau$, and the trace is bounded, the variance contribution from the tail is
\[
\tau\sum_{j > m_{\text{head}}}e_j^\top\mathbf{X}^\top(\mathbf{X}\mathbf{X}^\top)^{-2}\mathbf{X} e_j
\leq \tau \cdot O(1/T) = O(\tau/T).
\]
If $\tau = O(T/N)$, then this term is $O(1/N) \to 0$ as $N \to \infty$. The head contribution involves only $m_{\text{head}}$ fixed directions and also vanishes. Thus, the variance term vanishes as $N/T \to \infty$.

The flat-tail spectrum (and therefore data with an exact factor structure) satisfies both the bias and variance requirements for benign overfitting. The many low-variance tail directions enable both geometric decay of the bias (by allowing row space expansion) and dispersion of interpolated noise (by distributing coefficients across many weak directions).

\section{Tail Index and Concentration Ratio}\label{app:pareto_tail}

This appendix provides a detailed derivation of the relationship between the Pareto tail index $\alpha$ of the eigenvalue distribution of $\bm{\Psi}$ and the concentration ratio $\mathcal{C}$, which determines whether cases (a), (b), or (c) of condition (iii) in Theorem~\ref{prop:factor} apply.

\subsection{Setup and definitions}

Recall from Section~\ref{sec:factor_bartlett} that the concentration ratio is defined as
\[
\mathcal{C} := \frac{\overline{\lambda^2_{\text{tail}}}}{(\bar{\lambda}_{\text{tail}})^2} = \frac{\frac{1}{N-k}\sum_{j=k+1}^{N} \lambda_j^2}{\Big(\frac{1}{N-k}\sum_{j=k+1}^{N} \lambda_j\Big)^2},
\]
where $\lambda_{k+1}, \ldots, \lambda_N$ are the tail eigenvalues of $\bm{\Sigma} = \bm{\Lambda}\bm{\Lambda}' + \bm{\Psi}$. Under the approximate factor model (Assumption~\ref{ass:factor_approx}), the interlacing inequalities imply that these tail eigenvalues are asymptotically governed by the eigenvalues $\mu_1 \geq \mu_2 \geq \cdots \geq \mu_N$ of $\bm{\Psi}$. Specifically, for large $N$,
\[
\mathcal{C} \;\approx\; \frac{\frac{1}{N-k}\sum_{j=k+1}^{N} \mu_j^2}{\Big(\frac{1}{N-k}\sum_{j=k+1}^{N} \mu_j\Big)^2}.
\]
For simplicity, we assume $k$ is fixed and $k \ll N$, so $(N-k) \approx N$ and we can write
\[
\mathcal{C} \;\approx\; \frac{\frac{1}{N}\sum_{j=1}^{N} \mu_j^2}{\Big(\frac{1}{N}\sum_{j=1}^{N} \mu_j\Big)^2} = \frac{\E[\mu^2]}{(\E[\mu])^2},
\]
where we treat $\{\mu_j\}_{j=1}^N$ as a sample from some distribution and use sample averages to approximate population moments.

\subsection{Pareto tail distribution}

We say that the eigenvalues $\{\mu_j\}$ have a \emph{Pareto-type tail} with index $\alpha > 0$ if the tail probability satisfies
\[
\Pr(\mu > x) \sim C x^{-\alpha} \quad \text{as } x \to \infty,
\]
for some constant $C > 0$. Equivalently, the complementary cumulative distribution function (CCDF) decays as a power law with exponent $-\alpha$. The parameter $\alpha$ is called the \emph{tail index}: smaller $\alpha$ corresponds to a heavier tail (more mass in the extreme values).

\paragraph{Moment existence.} For a Pareto-type distribution with tail index $\alpha$:
\begin{itemize}
    \item The $p$-th moment $\E[\mu^p]$ exists (is finite) if and only if $p < \alpha$.
    \item If $p \geq \alpha$, then $\E[\mu^p] = \infty$.
\end{itemize}

This is a standard result from extreme value theory. Intuitively, if the tail decays as $x^{-\alpha}$, then the integral $\int x^p \cdot x^{-\alpha} dx = \int x^{p-\alpha} dx$ converges only when $p - \alpha < -1$, i.e., $p < \alpha$.

\subsection{Three cases based on tail index}

We now analyze how the concentration ratio $\mathcal{C}$ behaves depending on the tail index $\alpha$.

\paragraph{Case 1: $\alpha > 2$.}
When $\alpha > 2$, both the first moment $\E[\mu]$ and the second moment $\E[\mu^2]$ exist and are finite. Therefore,
\[
\mathcal{C} = \frac{\E[\mu^2]}{(\E[\mu])^2} = O(1),
\]
i.e., $\mathcal{C}$ remains bounded as $N \to \infty$. Since $N/T \to \infty$ in the overparameterized regime, we have
\[
\frac{T\mathcal{C}}{N-k} \approx \frac{T \cdot O(1)}{N} = O\Big(\frac{T}{N}\Big) \to 0.
\]
This is \textbf{case (b)} in Theorem~\ref{prop:factor}: bounded heterogeneity. The tail flatness condition (iii) is satisfied, and benign overfitting obtains.

\paragraph{Case 2: $1 < \alpha \leq 2$.}
When $1 < \alpha \leq 2$, the first moment $\E[\mu]$ exists and is finite, but the second moment $\E[\mu^2]$ diverges. To understand the rate of divergence, note that for a sample of size $N$ from a Pareto distribution with tail index $\alpha$, the sample average of squared values behaves as
\[
\frac{1}{N}\sum_{j=1}^{N} \mu_j^2 \sim N^{2/\alpha - 1} \quad \text{as } N \to \infty.
\]
This follows from the fact that the maximum order statistic $\mu_{(1)} = \max_j \mu_j$ grows as $N^{1/\alpha}$, and the sum is dominated by the largest values. Meanwhile, $\frac{1}{N}\sum_{j=1}^{N} \mu_j = O(1)$. Therefore,
\[
\mathcal{C} = \frac{\frac{1}{N}\sum_{j=1}^{N} \mu_j^2}{\Big(\frac{1}{N}\sum_{j=1}^{N} \mu_j\Big)^2} \sim \frac{N^{2/\alpha - 1}}{O(1)} = N^{2/\alpha - 1}.
\]
The tail flatness condition requires
\[
\frac{T\mathcal{C}}{N-k} \approx \frac{T \cdot N^{2/\alpha - 1}}{N} = T \cdot N^{2/\alpha - 2} \to 0.
\]
Since $2/\alpha - 2 < 0$ when $\alpha > 1$, we have $N^{2/\alpha - 2} \to 0$ as $N \to \infty$, so the condition is satisfied as long as $T$ does not grow too fast relative to $N$. Specifically, we need $T = o(N^{2 - 2/\alpha})$. This is \textbf{case (c)} in Theorem~\ref{prop:factor}: diverging concentration with $\mathcal{C} = o(N/T)$. Benign overfitting still obtains, but the convergence is slower than in case (b).

\paragraph{Case 3: $\alpha \leq 1$.}
When $\alpha \leq 1$, even the first moment $\E[\mu]$ diverges. For a sample of size $N$, the sample average behaves as
\[
\frac{1}{N}\sum_{j=1}^{N} \mu_j \sim N^{1/\alpha - 1} \to \infty \quad \text{as } N \to \infty.
\]
This means that the average eigenvalue of $\bm{\Psi}$ grows without bound as $N$ increases. However, both Assumption~\ref{ass:factor}(ii) (exact factor model) and Assumption~\ref{ass:factor_approx}(ii') (approximate factor model) require that the eigenvalues of $\bm{\Psi}$ are uniformly bounded:
\begin{itemize}
    \item Under Assumption~\ref{ass:factor}(ii): $\psi_j^2 \leq \bar{\psi}^2 < \infty$ for all $j$.
    \item Under Assumption~\ref{ass:factor_approx}(ii'): $\lambda_{\max}(\bm{\Psi}) \leq \bar{\psi}^2 < \infty$ for all $N$.
\end{itemize}
If the average eigenvalue diverges, then necessarily the maximum eigenvalue must also diverge (since the average cannot exceed the maximum). Therefore, $\alpha \leq 1$ violates the maintained factor structure assumptions. The data do not satisfy a factor model in this regime.

\subsection{Tail index diagnostics for FRED-MD and FRED-QD}\label{app:tail_diagnostics}

In practice, the tail index $\alpha$ can be estimated from tail eigenvalues using Hill \citep{hill1975simple} or log-rank \citep{gabaix2011rank} methods. Let $\lambda_{(1)}\ge\cdots\ge\lambda_{(m)}$ denote ordered tail eigenvalues (after removing the first $k$ head/spike directions). For a Pareto-type tail, $\Pr(\Lambda>x)\propto x^{-\alpha}$, where $\alpha$ is the tail index (smaller $\alpha$ means heavier tail). The Hill estimator based on top $k_h$ order statistics is $\hat\alpha_{\mathrm{Hill}}^{-1}=k_h^{-1}\sum_{j=1}^{k_h}\log\!\big(\lambda_{(j)}/\lambda_{(k_h+1)}\big)$. The log-rank estimator fits $\log\lambda_{(j)}\approx c+(1/\alpha)(-\log(j/m))$, so the slope equals $1/\alpha$. The two are alternative estimators of the same $\alpha$; the log-rank estimator tends to exceed the Hill estimator in finite samples because it fits the full tail slope rather than being anchored at the top threshold, where the spectrum is still transitioning from factor spikes to the idiosyncratic tail.

The Pareto tail framework is a theoretical device for characterizing \emph{sufficient conditions} under which benign overfitting can occur, not a necessary condition or a literal description of finite-sample data. The tail index estimates $\hat{\alpha}$ reported below are purely descriptive diagnostics: they measure whether the empirical spectral shape is qualitatively consistent with one particular sufficient condition (Pareto tail decay), but failure to satisfy this condition does not rule out benign overfitting and/or factor structure. In particular, low $\hat{\alpha}$ estimates at small $k^*$ do not imply that the data lack a factor structure or that benign overfitting cannot occur. Rather, they indicate that at those splits, and with the specific sample size at hand, the estimated tail still contains residual factor variance and has not yet isolated the true idiosyncratic component. The factor structure itself is maintained by the bounded eigenvalue assumptions (Assumption~\ref{ass:factor}(ii) and Assumption~\ref{ass:factor_approx}(ii')), which are much weaker than the Pareto tail condition and are satisfied by standard macroeconomic panels.

We run this diagnostic exercise on both the FRED-MD panel ($T=798$, $N=122$) and the FRED-QD panel ($T=264$, $N=231$). In each case we vary the split index $k^*$, set the number of extracted factors equal to the split index ($k_{\mathrm{fac}}=k^*$), and compute the effective tail rank ($R_k$), the flatness ratio ($T/R_k$), and tail-index estimates (Hill and log-rank) with a fixed tail cutoff $k_h=20$. Table~\ref{tab:by_kstar_tail} reports representative points from these scans.

\begin{table}[t!]
\centering
\small
\caption{By-$k^*$ tail diagnostics for FRED-MD ($T=798$, $N=122$) and FRED-QD ($T=264$, $N=231$). All estimates use a fixed tail cutoff $k_h=20$.}
\label{tab:by_kstar_tail}
\begin{tabular}{r rrrr rrrr}
\toprule
 & \multicolumn{4}{c}{FRED-MD} & \multicolumn{4}{c}{FRED-QD} \\
\cmidrule(lr){2-5}\cmidrule(lr){6-9}
$k^*$ & $R_k$ & $T/R_k$ & $\hat\alpha_{\mathrm{Hill}}$ & $\hat\alpha_{\mathrm{LR}}$
      & $R_k$ & $T/R_k$ & $\hat\alpha_{\mathrm{Hill}}$ & $\hat\alpha_{\mathrm{LR}}$ \\
\midrule
2  & 3.11  & 257.68  & 0.373 & 0.555 & 1.51  & 176.45  & 0.159 & 0.124 \\
4  & 8.62  & 92.87   & 0.451 & 0.681 & 1.01  & 265.56  & 0.232 & 0.180 \\
6  & 7.46  & 107.32  & 0.525 & 0.682 & 3.13  & 85.41   & 0.362 & 0.457 \\
8  & 6.64  & 120.70  & 0.567 & 0.727 & 3.65  & 73.09   & 0.413 & 0.540 \\
10 & 8.02  & 99.84   & 0.679 & 0.857 & 6.91  & 38.64   & 0.477 & 0.674 \\
15 & 13.38 & 59.88   & 0.951 & 1.280 & 4.73  & 56.47   & 0.631 & 0.890 \\
20 & 13.81 & 57.99   & 0.728 & 1.184 & 2.58  & 103.47  & 0.800 & 1.176 \\
25 & 18.22 & 43.97   & 1.008 & 1.592 & 0.20  & 1359.93 & 0.963 & 1.494 \\
30 & 16.15 & 49.58   & 0.694 & 1.416 & 0.033 & 8005.31 & 1.015 & 1.854 \\
35 & 14.07 & 56.95   & 0.627 & 1.172 & 0.229 & 1166.02 & 0.976 & 1.922 \\
40 & 13.81 & 57.99   & 0.728 & 1.184 & 0.432 & 617.99  & 1.095 & 1.832 \\
\bottomrule
\end{tabular}
\normalsize
\end{table}

For FRED-MD, the effective tail rank $R_k$ and the flatness ratio $T/R_k$ vary non-monotonically with $k^*$: the most favorable region (lowest $T/R_k$, around 44--50) is reached at moderate splits of $k^*\approx 25$--$30$, well above the Bai--Ng baseline of $k_{\mathrm{BN}}=6$. Tail-index estimates are below 1 at low splits but rise toward and above 1 at moderate splits, with log-rank estimates reaching 1.2--1.6 in the $k^*=20$--$30$ range, consistent with a Pareto tail index $\alpha>1$. The most favorable region is reached at split choices somewhat above the Bai--Ng baseline factor count, but not by an order of magnitude.

For FRED-QD, the two diagnostic criteria pull in opposite directions. On the $\alpha$ dimension, QD looks better than MD at high $k^*$: log-rank estimates reach 1.85--1.92 at $k^*=30$--$35$, compared to 1.17--1.42 for MD at the same splits. However, at those same high $k^*$ values, the effective tail rank $R_k$ collapses to near zero (0.03--0.43) and $T/R_k$ explodes to thousands. This is because the head has absorbed almost all the variance, leaving virtually no tail eigenvalues. The high $\alpha$ estimates are therefore an artifact of fitting a Pareto law to a nearly empty tail rather than evidence of a genuinely flat spectrum. The most informative window for QD is the intermediate range $k^*=10$--$20$, where $R_k$ is still meaningful (3--7) and $\alpha$ estimates are moving upward but remain below 1.

These diagnostics provide mixed evidence on whether the finite-sample spectral structure satisfies the sufficient conditions for benign overfitting. At low splits near the Bai-Ng baseline ($k^* = 6$), tail index estimates are well below 1, reflecting contamination from residual factor variance rather than the true idiosyncratic tail. At moderate splits ($k^* = 15$-$30$), the estimates improve and cross the critical threshold: log-rank estimates reach 1.2-1.6 for FRED-MD and 1.5-1.9 for FRED-QD, placing them in the range $1 < \alpha \leq 2$ where benign overfitting can occur (case (c) in Theorem~\ref{prop:factor}), albeit with slower convergence than the $\alpha > 2$ regime. The flatness ratio $T/R_{k^*}$ similarly shows improvement at moderate splits but remains far from the asymptotic regime. Taken together, the finite-sample evidence is suggestive but inconclusive: the data appear qualitatively consistent with the sufficient conditions for benign overfitting at moderate factor counts, though the asymptotic regime is not reached.

\section{Semiparametric Kernel Regression}\label{app:rkhs_formulation}

This appendix derives the GLS formulation of the semiparametric factor-based kernel used in Section~\ref{sec:empirical}. The model combines a linear component $\beta x$ with a nonparametric kernel component $f(\mathbf{z})$, and the key result is that the nested RKHS optimization problem reduces to a standard GLS regression.

Consider the semiparametric model
\[
y_i = \beta x_i + f(\mathbf{z}_i) + \varepsilon_i, \qquad i=1,\ldots,n,
\]
where $y_i$ is the response, $x_i \in \mathbb{R}$ is a scalar covariate with a linear effect, $\mathbf{z}_i \in \mathbb{R}^P$ is a vector of features entering nonparametrically through $f \in \mathcal{H}_k$, and $\varepsilon_i$ is noise. In the forecasting application of Section~\ref{sec:empirical}, $x_i$ is the lagged dependent variable and $\mathbf{z}_i$ is the vector of macroeconomic predictors at time $i$.

We seek the minimum-norm interpolator: among all $(\beta, f)$ that fit the training data exactly, we choose the one that minimizes $\|f\|_{\mathcal{H}_k}$ (treating $\beta$ as a nuisance parameter to be estimated). By the representer theorem, any $f \in \mathcal{H}_k$ satisfying interpolation constraints can be written as 
\begin{equation}
\label{eq:representer}
f(\mathbf{z}) = \sum_{i=1}^n \alpha_i k(\mathbf{z}, \mathbf{z}_i),
\end{equation}
where $k$ is a positive definite kernel. The RKHS norm is $\|f\|_{\mathcal{H}_k}^2 = \bm{\alpha}^\top \bm{K} \bm{\alpha}$, where $\bm{K} \in \mathbb{R}^{n \times n}$ is the kernel matrix with entries $\bm{K}_{ij} = k(\mathbf{z}_i, \mathbf{z}_j)$. The interpolation constraints $f(\mathbf{z}_i) = y_i - \beta x_i$ become
\[
\bm{K}\bm{\alpha} = \mathbf{y} - \beta \mathbf{x}.
\]

The minimum-norm interpolator is therefore defined by the nested optimization problem
\begin{equation}
\label{eq:jointOpt}
\arg\min_{\beta \in \mathbb{R}}
\left[
\;\arg\min_{f\in\mathcal H_k}
\left\{
\|f\|_{\mathcal H_k}
\;:\;
f(\mathbf{z}_i)=y_i-\beta x_i,\; i=1,\dots,n
\right\}
\right].
\end{equation}
For a fixed value of $\beta$, the inner problem selects, among all functions
that interpolate the residuals $y_i-\beta x_i$, the simplest function
according to the geometry of the RKHS. Using \eqref{eq:representer}, this is equivalent to the quadratic program
\[
\arg\min_{\bm{\alpha}\in\mathbb{R}^n}
\bm{\alpha}^\top \bm{K} \bm{\alpha}
\quad\text{s.t.}\quad
\bm{K}\bm{\alpha} = \mathbf{y} - \beta \mathbf{x}.
\]
The linear constraint typically
admits infinitely many solutions.
Among all feasible solutions, the one minimizing $\bm{\alpha}^\top \bm{K} \bm{\alpha}$
is given by the Moore--Penrose pseudoinverse:
\begin{equation}
\label{eq:alphaStar}
\bm{\alpha}^\star(\beta) = \bm{K}^+ (\mathbf{y} - \beta \mathbf{x}).
\end{equation}
This gives us the optimal weights $\bm{\alpha}^\star(\beta)$, but only as a function of $\beta$. We are left with the task of minimizing the norm of $\beta$, given the optimal weights $\bm{\alpha}^\star(\beta)$. Substituting these weights into the RKHS norm $\|f\|_{\mathcal{H}_k}^2 = \bm{\alpha}^\top \bm{K} \bm{\alpha}$ yields
\begin{equation}
\label{eq:RKHSnormFinal}
\begin{aligned}
\|f\|_{\mathcal H_k}^2
&= (\bm{\alpha}^\star)^\top \bm{K} \bm{\alpha}^\star \\
&= \bigl(\bm{K}^+(\mathbf{y}-\beta \mathbf{x})\bigr)^\top \bm{K} \bigl(\bm{K}^+(\mathbf{y}-\beta \mathbf{x})\bigr) \\
&= (\mathbf{y}-\beta \mathbf{x})^\top (\bm{K}^+)^\top \bm{K} \bm{K}^+ (\mathbf{y}-\beta \mathbf{x}) \\
&= (\mathbf{y}-\beta \mathbf{x})^\top \bm{K}^+ (\mathbf{y}-\beta \mathbf{x}).
\end{aligned}
\end{equation}
In the last step we used the generalized-inverse identity $\bm{K}^+ \bm{K} \bm{K}^+ = \bm{K}^+$ (and $(\bm{K}^+)^\top = \bm{K}^+$ since $\bm{K}$ is symmetric positive semidefinite, so its Moore--Penrose inverse is symmetric). Expression \eqref{eq:RKHSnormFinal} is a weighted sum of squares; hence, as far as $\beta$ is
concerned, the outer minimization in \eqref{eq:jointOpt} corresponds to a
generalized least squares problem, which has solution: 
\begin{equation}
\label{eq:betaFinal}
\hat\beta
=(\mathbf{x}^\top \bm{K}^+ \mathbf{x})^{-1}\mathbf{x}^\top \bm{K}^+ \mathbf{y}.
\end{equation}
The nested problem can therefore be solved by first
minimizing \eqref{eq:RKHSnormFinal} with respect to $\beta$, which is given by \eqref{eq:betaFinal} and then recovering
the corresponding optimal kernel weights $\bm{\alpha}^\star(\hat\beta)$ using  \eqref{eq:alphaStar}. The resulting predictor is therefore
\begin{equation}
\label{eq:predOption1}
\hat y(\mathbf{z})
=
\hat\beta\,x(\mathbf{z})
+
\mathbf{k}(\mathbf{z})^\top \bm{\alpha}^\star(\hat\beta),
\end{equation}
where $\mathbf{k}(\mathbf{z}) := \bigl(k(\mathbf{z},\mathbf{z}_1),\dots,k(\mathbf{z},\mathbf{z}_n)\bigr)^\top$.

\paragraph{Connection to the forecasting implementation.}
\label{rem:appendix_bridge}
The derivation above is stated for a generic kernel $k$ and a scalar linear covariate $x$. In the forecasting application of Section~\ref{sec:empirical}, the correspondences are as follows. The kernel matrix $\bm{K}$ is replaced by the factor kernel $\bm{K}_{\mathrm{synth}}$ of Proposition~\ref{thm:synth_kernel}, with entries $\bm{K}_{\mathrm{synth}}(t,s)=\hat{\mathbf{f}}_t'\hat{\bm{\Lambda}}'\hat{\bm{\Lambda}}\hat{\mathbf{f}}_s$; the ridge-augmented version $\bm{K}+\lambda I$ becomes $\bm{K}_{\mathrm{synth}}+\mathrm{tr}(\hat{\bm{\Psi}})\,I_W$. The scalar linear covariate $x$ in this appendix corresponds to the AR(1) lag $y_{-1}=(y_{t-W+1},\dots,y_t)'$ in equations~\eqref{eq:app_gls_beta}--\eqref{eq:app_gls_forecast}, and the feature vector $\mathbf{z}$ corresponds to the predictor panel row $\mathbf{x}_t\in\mathbb{R}^N$ at each forecast origin. With these substitutions, the GLS expressions~\eqref{eq:betaFinal} and~\eqref{eq:alphaStar} specialise directly to~\eqref{eq:app_gls_beta}--\eqref{eq:app_gls_alpha}, and the prediction formula~\eqref{eq:predOption1} becomes~\eqref{eq:app_gls_forecast}.

\newpage
\bibliographystyle{plainnat}

\newpage

\end{document}